\documentclass[hidelinks,onefignum,onetabnum,final]{siamart220329}

\usepackage{lipsum}
\usepackage{amsfonts}
\usepackage{graphicx}
\usepackage{epstopdf}
\usepackage{algorithmic}
\usepackage{enumitem}
\usepackage{amsmath} 
\usepackage{amssymb}
\usepackage{mathtools}
\usepackage{esint}
\usepackage{mathrsfs}
\ifpdf
  \DeclareGraphicsExtensions{.eps,.pdf,.png,.jpg}
\else
  \DeclareGraphicsExtensions{.eps}
\fi

\newsiamremark{remark}{Remark}
\newsiamremark{hypothesis}{Hypothesis}
\newsiamremark{example}{Example}
\crefname{hypothesis}{Hypothesis}{Hypotheses}
\newsiamthm{claim}{Claim}

\title{A Mathematical Aspect of Bloch's Theorem\thanks{Submitted to the editors DATE.
\funding{This work was supported by the National Key R \& D Program of China under grants 2019YFA0709600 and 2019YFA0709601, the National Natural Science Foundation of China under grant 12021001.}}}

\author{Yan Li\thanks{LSEC, Institute of Computational Mathematics and Scientific/Engineering Computing, Academy of Mathematics and Systems Science, Chinese Academy of Sciences, Beijing 100190, China; and School of Mathematical Sciences, University of Chinese Academy of Sciences, Beijing 100049, China (\email{liyan2021@lsec.cc.ac.cn}, \email{azhou@lsec.cc.ac.cn}).}
\and Bin Yang\thanks{NCMIS, Academy of Mathematics and Systems Science, Chinese Academy of Sciences, Beijing 100190, China (\email{binyang@lsec.cc.ac.cn}).}
\and Aihui Zhou\footnotemark[2]}

\usepackage{amsopn}

\ifpdf
\hypersetup{
  pdftitle={A Mathematical Aspect of Bloch's Theorem},
  pdfauthor={Y. Li, B.Yang, and A.Zhou}
}
\fi

\begin{document}

\maketitle

% REQUIRED
\begin{abstract}
In this paper, by studying a class of 1-D Sturm-Liouville problems with periodic coefficients, we show and classify the solutions of periodic Schrodinger equations in a multidimensional case, which tells that not all the solutions are Bloch solutions. In addition, we also provide several properties of the solutions and quasimomenta and illustrate the relationship between bounded solutions and Bloch solutions.
\end{abstract}

% REQUIRED

% REQUIRED

\section{Introduction}
 Bloch's theorem  provides a framework for describing the behavior of electrons in crystal lattices, which is essential for predicting and interpreting the properties of materials used in various technological applications. It provides a powerful tool for solving partial differential equations in periodic media which are common in scientific and engineering computing. By exploiting the periodic feature of these problems, the theorem allows us to reduce the complexity of the solution space and achieve highly accurate numerical approximations. In particular, according to Bloch's theorem, we can discretize periodic systems by the plane wave method, which is by far most used in present-day electronic structure calculations and  is the most natural and popular method for periodic solids calculations. We refer to \cite{ ashcroft1976solid, kittel2018introduction, kuchment1993floquet, kuchment2016overview,  martin2020electronic, reed1978methods, singh2006planewaves, sjostedt2000alternative} for more details.

However, Kuchment pointed in \cite{kuchment2016overview} that Bloch's theorem was probably never proved by Bloch himself. We understand that there are some versions of Bloch's theorems and proofs in literatures. A periodic boundary based version stated and proved in \cite{ashcroft1976solid} reads as:  {\it the eigenstate $\Psi$ of one-electron Hamiltonian $\mathscr{H}=-\frac{\hbar^2}{2m}\Delta+V(\mathbf{r})$, where $V(\mathbf{r}+a_{\gamma})=V(\mathbf{r})$ for all $a_{\gamma}$ in the  Bravais lattice $\mathcal{R}$, can be chosen to have the form of a plane wave times a function with the periodicity of the Bravais lattice: $\Psi_{n\mathbf{k}}(\mathbf{r})=e^{\mathrm{i}\mathbf{k}\cdot\mathbf{r}}u_{n\mathbf{k}}(\mathbf{r})$, where $u_{n\mathbf{k}}(\mathbf{r}+a_{\gamma})=u_{n\mathbf{k}}(\mathbf{r})$ for all $a_{\gamma}$ in the Bravais lattice.}  Solutions of the form as $\Psi_{n\mathbf{k}}(\mathbf{r})=e^{\mathrm{i}\mathbf{k}\cdot\mathbf{r}}u_{n\mathbf{k}}(\mathbf{r})$  are generally called the Bloch solutions. The vectors $\mathbf{k}$ is called quasimomenta or crystal momenta.

To prove the theorem, a periodic (or Born-von Karman) boundary condition as follows (see, e.g.,   pages 133-139 of \cite{ashcroft1976solid})
\begin{equation}
    \Psi(\mathbf{r}+N_{i}\mathbf{a}_{i})=\Psi(\mathbf{r}), i=1,2,3,
\end{equation}
is applied, where $\mathbf{a_{i}}$ are primitive vectors of the Bravais lattice and $N_{i}$ are all integers.  $N=N_{1}N_{2}N_{3}$ is the total number of primitive cells in the crystal.

Similar version of the proof is also provided by Dresselhaus et al.
(see Chapter 13 of \cite{dresselhaus2008applications} for details).  We observe that, for a general Bloch solution $\Psi(\mathbf{r})=e^{\mathrm{i}\mathbf{k}\cdot\mathbf{r}}u(\mathbf{r})$, such type of boundary conditions might not be satisfied.  For instance, if for some $i$, $\mathbf{k}\cdot\mathbf{a}_{i}$ and $2\pi$ are incommensurable (in other words, $\frac{\mathbf{k}\cdot\mathbf{a}_{i}}{2\pi}$ equals an irrational number), there will never exist any integer $M$ such that $\Psi(\mathbf{r}+M\mathbf{a}_{i})=\Psi(\mathbf{r}),\,\forall \mathbf{r}\in\mathbb{R}^{3}$. Consequently the proofs with the above periodic boundary condition are not applicable to Bloch's theorem for all possible quasimomenta.

Expect for the periodic boundary condition, the typical proof is also  based on the so-called fundamental theorem of quantum mechanics (see, e.g., Page 144 of \cite{park1974introduction}, \cite{kramers2018quantum}) as follows:  {\it Let $\hat{P}, \hat{Q}, \hat{R}, \dots$ be a set of Hermitian operators which commute among each other. Then a complete set of solutions that are simultaneously solutions of all the operators can be found}. With the commutativity of $\mathscr{H}$ and translation operators $T_{a_{\gamma}}:f(\mathbf{r})\mapsto f(\mathbf{r}+a_{\gamma})$ for $\forall a_{\gamma}\in\mathcal{R}$, the eigenstates of $\mathscr{H}$ can therefore be chosen to be simultaneous eigenstates of $T_{a_{\gamma}}$ and then the proof is finished \cite{ashcroft1976solid, kittel2018introduction}. Unfortunately, on one hand, it is not addressed in the above books on what Hilbert spaces $T_{a_{\gamma}}$ and $\mathscr{H}$ are defined, which is essential for Hermitian operators. On the other hand, an example as follows shows that the set of simultaneous eigenstates of $\mathscr{H}$ and $T_{a_{\gamma}}$ might be a proper subset of eigenstates of $\mathscr{H}$.  In the one-dimensional case, if we set the potential as the piecewise function 
\begin{equation*}
    V(r)=\frac{\hbar^{2}}{2m}(\frac{2}{(r-n-1)^{2}}-\frac{2}{r-n-1}), \quad n-1\leqslant r<n, n\in\mathbb{Z},
\end{equation*}
 then the piecewise function 
 \begin{equation*}
     \psi(r)=e^{r}(1-\frac{r}{r-n-1}), n-1\leqslant r<n, n\in\mathbb{Z}
 \end{equation*} is an eigenstate of $\mathscr{H}$ corresponding to $E=-\frac{\hbar^{2}}{2m}$ while $\psi(r)$ is not a Bloch solution and hence is not the solution of $T_{a_{\gamma}}$ either. 

We also observe that there is another version of Bloch’s theorem and proofs that implies the existence of Bloch solutions in \cite{kuchment1993floquet, kuchment2016overview, cycon2009schrodinger, glazman1965direct,  shnol1957behavior, shubin1989weak}, that is,  \it the existence of a nontrivial bounded solution of a periodic elliptic equation $\mathscr{H}\Psi=\lambda \Psi$ implies the existence of a Bloch solution with a real quasimomentum, and thus $\lambda\in\sigma(\mathscr{H})$. \rm Here $\sigma(\mathscr{H})$ is the spectrum of $\mathscr{H}$. However, whether the solutions are all Bloch solutions have not been confirmed.

Consequently, we arrive at the conclusion that the typical versions of Bloch's theorem may confirm a proper subset of eigenstates of $\mathscr{H}$ are Bloch solutions only. It is unclear if all the eigenstates are Bloch solutions.

Besides the above versions concentrating on solutions, the spectrum of periodic Schr\"{o}dinger operators has been also discussed \cite{kuchment1993floquet, kuchment2016overview, cycon2009schrodinger, eastham1973spectral, gontier2016convergence,  Yakubovich1975linear, reed1980methods, gesztesy2009criterion, magnus2013hill, titchmarsh2011elgenfunction}. Here, the Hamiltonian $\mathscr{H}$ operates  over $L^{2} (\mathbb{R}^{d})$ with the domain $H^{2}(\mathbb{R}^{d})(d\geqslant1)$. With the Bloch-Floquet transform, the operator $\mathscr{H}$ over $L^{2} (\mathbb{R}^{d})$ can be decomposed into operators over $L^{2} (U)$, where $U$ is the individual unit cell; or in other words, $\mathscr{H}$ can be diagonalized. Then it shows that
\it for $\lambda\in\sigma(\mathscr{H})$, there exists a Bloch function $\Psi$ satisfying $\mathscr{H}\Psi=\lambda \Psi$. \rm However, in fact, $\mathscr{H}$  has only absolutely continuous spectrum with no eigenvalues \cite{reed1980methods}. We should name such $\Psi$ generalized eigenfunctions \cite{kuchment2016overview} and set $\sigma_{g}(\mathscr{H})=\{\lambda\in\mathbb{R}|\mathscr{H}\Psi=\lambda\Psi, \Psi\not\equiv0\}$ . It is clear that  $\sigma(\mathscr{H})\subset\sigma_{g}(\mathscr{H})$. It is quite interesting to see whether there holds $\sigma(\mathscr{H})=\sigma_{g}(\mathscr{H})$. If not, what are the characteristics of the set $\sigma_{g}(\mathscr{H})\setminus\sigma(\mathscr{H})$?

In this paper, we try to answer the above questions, fill in the gap, and provide a mathematical aspect on the Bloch theorem. Motivated by the approach of separation of variables for solving partial differential equations, we provide a classification of solutions of periodic Schrodinger equations when they are separable solutions whose variables can be separated directly or with some linear mapping in multidimensions. We present a general form of such type of solutions and show their properties, and illustrate the relationship between bounded solutions and Bloch solutions.
 
 The rest of this paper is organized as follows. The Bloch-Floquet transform is introduced in Section 2. In Section 3, by studying a class of 1-D Sturm-Liouville problems with periodic coefficients, we present a general form of separable solutions and show their properties, and illustrate the relationship between bounded solutions and Bloch solutions. Finally, we give some concluding remarks in Section 4 and provide more details about one dimensional cases in \ref{appdix:class-eigenfunc}.

\section{Bloch-Floquet transform}
Let $\gamma_{1}, \gamma_{2}, \cdots, \gamma_{d}\in\mathbb{R}^{d}$ be $d$  linearly independent vectors. As a result, $\operatorname{det}(\gamma_{1},\gamma_{2},\ldots,\gamma_{d})\ne 0$. The Bravais lattice \cite{martin2020electronic, lewin2010spectral}  is defined as
\begin{displaymath}
  \mathcal{R}:=\gamma_{1}\mathbb{Z}\oplus\gamma_{2}\mathbb{Z}\oplus\cdots\oplus\gamma_{d}\mathbb{Z}.
\end{displaymath}

Then, the dual (reciprocal) lattice to $\mathcal{R}$ is defined as 
\begin{displaymath}
  \mathcal{R}^{*}:=\gamma_{1}^{*}\mathbb{Z}\oplus\gamma_{2}^{*}\mathbb{Z}\oplus\cdots\oplus\gamma_{d}^{*}\mathbb{Z},
\end{displaymath}
where $\gamma_{j}^{*}\gamma_{i}=2\pi\delta_{ij}$. And the Brillouin zone \cite{lewin2010spectral} is 
\begin{displaymath}
  \mathcal{B}:=\{\mathbf{k}\in\mathbb{R}^{d}\,|\,\Vert \mathbf{k}\Vert=\inf_{b_{\gamma}\in\mathcal{R}^{*}}\Vert \mathbf{k}-b_{\gamma}\Vert\}.
\end{displaymath}

    A function $\Psi$ defined on $\mathbb{R}^{d}$ is called $\mathcal{R}$-periodic if 
    \begin{displaymath}
        \Psi(\mathbf{r}+a_{\gamma})=\Psi(\mathbf{r}),~\forall a_{\gamma}\in\mathcal{R},~\forall \mathbf{r}\in\mathbb{R}^{d}.
    \end{displaymath}

In this paper, we shall use the standard notation for the Sobolev space $H^{s}(\mathbb{R}^{d})$ with associated norms (see, e.g. \cite{adams2003sobolev}, \cite{ciarlet1990handbook} and \cite{evans2022partial}). 

We consider the spectrum of the one-electron Schrödinger operator with a periodic potential: 
\begin{equation*}
   \mathscr{H}:=-\frac{\hbar^{2}}{2m_{e}}\Delta+V, 
\end{equation*}
where $m_{e}$ is the mass of electron,
$\hbar$ is the Planck constant over $2\pi$, and 
the potential $V$ is $\mathcal{R}$-periodic. Here, $\mathscr{H}$ operates over $L^{2}(\mathbb{R}^{d})$ with the domain $H^{2}(\mathbb{R}^{d})$. Then we introduce the Bloch-Floquet transform, with which the operator 
$\mathscr{H}$ defined over $\mathbb{R}^{d}$ can be decomposed into operators on individual unit cells, or in other words, $\mathscr{H}$ can be diagonalized.

Set the unit cells of $\mathcal{R}$ and $\mathcal{R}^{*}$:
\begin{equation*}
\begin{array}{rcl}
     \mathbb{T}:=&\{\sum_{i=1}^{d}\mu_{i}\gamma_{i}:\mu_{i}\in[-1,1), i=1,\cdots,d\},\\ \mathbb{T}^{*}:=&\{\sum_{i=1}^{d}\mu_{i}\gamma_{i}^{*}:\mu_{i}\in[-1,1), i=1,\cdots,d\}.
\end{array}
\end{equation*}
\begin{definition}[\cite{kuchment2016overview, gontier2016convergence}]\label{def:b_f_trans}
The Bloch-Floquet transform $\mathcal{U}_{\mathcal{R}}$ acts as follows:
\begin{equation*}
    f(\mathbf{r})\mapsto\mathcal{U}_{\mathcal{R}}f(\mathbf{r},\mathbf{k})\equiv\sum_{a_{\gamma}\in\mathcal{R}}f(\mathbf{r}+a_{\gamma})e^{-\mathrm{i}\mathbf{k}\cdot a_{\gamma}}, \forall f\in L^{2}(\mathbb{R}^{d}).
\end{equation*}
For convenience, we set $f_{\mathbf{k}}:=\mathcal{U}_{\mathcal{R}}f(\cdot,\mathbf{k})$.
\end{definition}

$\mathcal{U}_{\mathcal{R}}f(\mathbf{r},\mathbf{k})$ is just the sum of the Fourier series of $f(\mathbf{r})$, whose coefficients are the pieces of $f$ over the shifted pieces $\mathbb{T}+a_{\gamma}$. Due to the fact $\mathbb{R}^{d}=\bigcup_{a_{\gamma}\in \mathcal{R}}(\mathbb{T}+a_{\gamma})$, we see that Bloch-Floquet transform is indeed to decompose $f$ over $\mathbb{R}^{d}$ into unit cells. 

It is easy to observe the following useful properties:
\begin{equation}\label{eq:ues_pro1}
\begin{array}{rcl}
     \mathcal{U}_{\mathcal{R}}f(\mathbf{r}+a_{\gamma},\mathbf{k})&=&e^{\mathrm{i}\mathbf{k}\cdot a_{\gamma}}\mathcal{U}_{\mathcal{R}}f(\mathbf{r},\mathbf{k}),\quad a_{\gamma}\in\mathbb{T},\\
    \mathcal{U}_{\mathcal{R}}f(\mathbf{r},\mathbf{k}+b_{\nu})&=&\mathcal{U}_{\mathcal{R}}f(\mathbf{r},\mathbf{k}), \quad b_{\nu}\in\mathbb{T}^{*}.
\end{array}
\end{equation}
\cref{eq:ues_pro1} implies that $\mathcal{U}_{\mathcal{R}}f(\mathbf{r},\mathbf{k})$ adjusts the amplitude periodically involving $\mathbf{r}$ over $\mathcal{R}$ and is periodic involving $\mathbf{k}$ over $\mathcal{R}^{*}$.

We see from \cref{def:b_f_trans} that 
\begin{equation*}
    f(\mathbf{r})=\fint_{\mathbb{T}^{*}}\mathcal{U}_{\mathcal{R}}f(\mathbf{r},\mathbf{k})d\mathbf{k}, \quad \mathbf{r}\in\mathbb{R}^{d},
\end{equation*}
where $\fint_{\mathbb{T}^{*}}\equiv \frac{1}{|\mathbb{T}^{*}|}\int_{\mathbb{T}^{*}}$, which shows the inverse of $\mathcal{U}_{\mathcal{R}}$. Moreover, the following theorem holds for Bloch-Floquet transform  \cite{kuchment2016overview, gontier2016convergence}.
\begin{theorem}
Let $\bar{\mathbb{T}}$ be the closure of $\mathbb{T}$, then $\mathcal{U}_{\mathcal{R}}$ has the following properties:
    \begin{enumerate}
        \item  $\mathcal{U}_{\mathcal{R}}$ is isometry from $L^{2}(\mathbb{R}^{d})$ onto $L^{2}(\mathbb{T}^{*}, L^{2}(\bar{\mathbb{T}}))$;
        \item $\mathcal{U}_{\mathcal{R}}$ is isometry from $H^{s}(\mathbb{R}^{d})(s>0)$ onto a proper subspace of $L^{2}(\mathbb{T}^{*}, H^{s}(\bar{\mathbb{T}}))$;
        \item $\mathcal{U}_{\mathcal{R}}$ is isometry from $H^{s}(\mathbb{R}^{d})$ onto $L^{2}(\mathbb{T}^{*}, \mathcal{H}^{s})$, where $\mathcal{H}^{s}:=\bigcup_{\mathbf{k}\in\mathbb{T}^{*}}H_{\mathbf{k}}^{s}(\bar{\mathbb{T}})$, 
        \begin{equation*}
            H_{\mathbf{k}}^{s}(\bar{\mathbb{T}}):=\{f|_{\bar{\mathbb{T}}}|f\in H_{loc}^{s}(\mathbb{R}^{d}), f(\mathbf{r}+a_{\gamma})=e^{\mathrm{i}\mathbf{k}\cdot a_{\gamma}}f(\mathbf{r})\,  a.e., \forall a_{\gamma}\in\mathcal{R}\}.
        \end{equation*}
    \end{enumerate}
\end{theorem}

Let  $T_{a_{\gamma}}(a_{\gamma}\in\mathcal{R})$ be translation operators defined by
  \begin{equation*}
     T_{a_{\gamma}}f(\mathbf{r})=f(\mathbf{r}+a_{\gamma}),\quad \mathbf{r}\in\mathbb{R}^{d}.
 \end{equation*}

 Note that $\mathscr{H}T_{a_{\gamma}}=T_{a_{\gamma}}\mathscr{H}$ for $a_{\gamma}\in\mathcal{R}$. We recall the following existing results on direct integral decompositions \cite{kuchment2016overview, nielsen2020direct} of $\mathscr{H}$.  
\begin{theorem}
There hold that
    \begin{equation*}
         \mathcal{U}_{\mathcal{R}}(L^{2}(\mathbb{R}^{d}))=\fint_{\mathbb{T}^{*}}^{\oplus}L^{2}(\bar{\mathbb{T}}),\quad
         \mathcal{U}_{\mathcal{R}}(H^{s}(\mathbb{R}^{d}))=\fint_{\mathbb{T}^{*}}^{\oplus}H_{\mathbf{k}}^{s}(\bar{\mathbb{T}}),\quad
         \mathcal{U}_{\mathcal{R}}\mathscr{H}\mathcal{U}_{\mathcal{R}}^{-1}=\fint_{\mathbb{T}^{*}}^{\oplus}\mathscr{H}(\mathbf{k}),
    \end{equation*}
where $\mathscr{H}(\mathbf{k}):=\mathscr{H}|_{H_{\mathbf{k}}^{2}(\bar{\mathbb{T}})}$.
\end{theorem}

We see for $\mathbf{k}\in\mathbb{T}^{*}$ and $g\in H^{2}(\mathbb{R}^{d})$, that 
\begin{equation*}
    \mathcal{U}_{\mathcal{R}}\mathscr{H}\mathcal{U}_{\mathcal{R}}^{-1}(\mathcal{U}_{\mathcal{R}}g)(\cdot, \mathbf{k})=\mathscr{H}(\mathbf{k})(\mathcal{U}_{\mathcal{R}}g(\cdot,\mathbf{k})),
\end{equation*}
i.e., $(\mathscr{H}g)_{\mathbf{k}}=\mathscr{H}(\mathbf{k})g_{\mathbf{k}}$, which implies that $\mathcal{U}_{\mathcal{R}}\mathscr{H}\mathcal{U}_{\mathcal{R}}^{-1}$ is block diagonal involving $\mathbf{k}$.

The following spectral decomposition of $\mathscr{H}$ can be founded in \cite{kuchment2016overview, nielsen2020direct}.
\begin{theorem}\label{thm:spec_dec}
The spectrum $\sigma(\mathscr{H})$ is the union of the spectrum $\sigma(\mathscr{H}(\mathbf{k}))$ for $\mathbf{k}\in\mathbb{T}^{*}$, i.e.,   
\begin{equation*}
    \sigma(\mathscr{H})=\bigcup_{\mathbf{k}\in\mathbb{T}^{*}}\sigma(\mathscr{H}(\mathbf{k})).
\end{equation*}
\end{theorem}

\cref{thm:spec_dec} transforms the spectrum of $\mathscr{H}$ with the domain $H^{2}(\mathbb{R}^{d})$ into the spectrum of $\mathscr{H}(\mathbf{k})$ with the domain $H_{\mathbf{k}}^{2}(\bar{\mathbb{T}})$. We see that $\mathscr{H}$ is self-adjoint and bounded below, however, $\mathscr{H}$ has only absolutely continuous spectrum with no eigenvalues \cite{reed1980methods}. $\mathscr{H}(\mathbf{k})$ is also self-adjoint and bounded below, whose spectrum consists of real discrete point spectrum, i.e., 
\begin{equation*}
    \sigma(\mathscr{H}(\mathbf{k}))=\{\lambda_{j}(\mathbf{k})|\lambda_{1}(\mathbf{k})\leqslant \lambda_{2}(\mathbf{k})\leqslant \cdots \leqslant \lambda_{n}(\mathbf{k})\leqslant\cdots\}.
\end{equation*}

Relevant discussions are mentioned in \cite{kuchment2016overview}. We recall the following version of Bloch's theorem and provide a direct proof.

\begin{theorem}\label{thm:thm_bloch}
    For $\lambda\in\sigma(\mathscr{H})$, there exist $\mathbf{k}\in\mathbb{R}^{d}$ and the Bloch function $\Psi$ satisfying
\begin{equation}\label{eq2.1}
    \mathscr{H}\Psi=\lambda\Psi,
\end{equation}
  and
    \begin{equation}\label{eq:thm_bloch}
        \Psi(\mathbf{r}+a_{\gamma})=e^{\mathrm{i}\mathbf{k}\cdot a_{\gamma}}\Psi(\mathbf{r}), \quad a_{\gamma}\in\mathcal{R}, \quad \mathbf{r}\in\mathbb{R}^{d}.
    \end{equation}
\end{theorem}
\begin{proof}
Note that there exists $\mathbf{k}\in\mathbb{T}^{*}$ satisfying $\lambda\in\sigma(\mathscr{H}(\mathbf{k}))$ from \cref{thm:spec_dec}. Then there exists $\Psi\in H_{\mathbf{k}}^{2}(\bar{\mathbb{T}})$ such that 
\begin{equation*}
    \mathscr{H}\Psi=\lambda\Psi, \quad \mathbf{r}\in \mathbb{T}. 
\end{equation*}

Next, we show that $\Psi\in H_{\mathbf{k}}^{2}(\bar{\mathbb{T}})$ can be extended to a function over $\mathbb{R}^{d}$ and satisfy \cref{eq:thm_bloch}.

For $\mathbf{r}\in\mathbb{R}^{3}$, there exists $a_{\gamma}\in\mathcal{R}$ such that $\mathbf{r}+a_{\gamma}\in\mathbb{T}$. Then we arrive at 
\begin{equation*}
    \mathscr{H}\Psi(\mathbf{r})=T_{a_{\gamma}}^{-1}\mathscr{H}T_{a_{\gamma}}\Psi(\mathbf{r})=T_{a_{\gamma}}^{-1}\mathscr{H}\Psi(\mathbf{r}+a_{\gamma})=T_{a_{\gamma}}^{-1}\lambda\Psi(\mathbf{r}+a_{\gamma})=\lambda\Psi(\mathbf{r}),
\end{equation*}
which implies \cref{eq:thm_bloch}.
\end{proof}

The solution $\Psi$ satisfying \cref{eq:thm_bloch} is called the Bloch solution. Moreover, $\Psi$ can be reformulated as
\begin{equation*}
    \Psi(\mathbf{r})=e^{\mathrm{i}\mathbf{k}\cdot \mathbf{r}}p(\mathbf{r}), \quad\forall \mathbf{r}\in\mathbb{R}^{d},
\end{equation*}
where $\mathbf{k}\in\mathbb{C}^{d}$ is the quasimomentum and $p$ is $\mathcal{R}$-periodic.

Note that $\mathscr{H}:L^{2}(\mathbb{R}^{d})\rightarrow L^{2}(\mathbb{R}^{d})$ with the domain $H^{2}(\mathbb{R}^{d})$ has only absolutely continuous spectrum with no eigenvalues. We call the solution $\Psi$ of \cref{eq:thm_bloch} the generalized eigenfunction \cite{kuchment2016overview}. We set $\sigma_{g}(\mathscr{H})=\{\lambda\in\mathbb{R}|\mathscr{H}\Psi=\lambda\Psi, \Psi\not\equiv0\}$ and \cref{thm:thm_bloch} shows that $\sigma(\mathscr{H})\subset\sigma_{g}(\mathscr{H})$. It is quite interesting to see whether $\sigma(\mathscr{H})=\sigma_{g}(\mathscr{H})$. If not, what are the characteristics of the set $\sigma_{g}(\mathscr{H})\setminus\sigma(\mathscr{H})$? In the rest of the paper, by studying a class of 1-D Sturm-Liouville problems with periodic coefficients, we will show and classify the solutions of periodic Schrödinger equations in a multidimensional case.

\section{Classfication of solutions} 
Consider the $d$-D Schrödinger equation  with a $\mathcal{R}$-periodic potential: find $(\lambda, \Psi)$ satisfying
\begin{equation}\label{eq:ddddd}
    -\frac{1}{2}\Delta\Psi+V\Psi=\lambda\Psi.
\end{equation}

The classification of solutions of a Schr\"{o}dinger operator with a periodic potential over multidimensions is much more difficult.  Motivated by the approach of separation of variables for solving partial differential equations, we investigate a special case when variables of solutions can be separated directly or with some linear mapping in multidimensions.
At first, we study 1-D Sturm-Liouville problems with periodic coefficients that will be used.

\subsection{1-D Strum-Liouville problem}

For the one-dimensional case, we set $\mathcal{R}=\gamma\mathbb{Z}$ for some $\gamma\in\mathbb{R}$. The Sturm-Liouville problem reads: find $(\lambda, \psi)$ satisfying 
\begin{equation}\label{eq: sl}
    \mathscr{L}\psi:=-\frac{d^{2}\psi}{dx^{2}}+W\frac{d\psi}{dx}+V\psi=\lambda\psi, \lambda\in\mathbb{R},
\end{equation}
where $W(x)$ and $V(x)$ are $\mathcal{R}$-periodic, i.e.,  $W(x)=W(x+\gamma)$, $V(x)=V(x+\gamma)$. $V$ and $W$ can be complex valued.

We assume $W, V\in C_{p}^{0}(\mathbb{R})$ (see \cref{def:pc} of $C_{p}^{l}(\mathbb{R}), C_{p,\gamma}^{l}(\mathbb{R}) (l\geqslant0)$ in \ref{appdix:class-eigenfunc}). Since we discuss solutions of the Sturm-Liouville problem in the classical sense, the search for solutions is primarily restricted to the class $C_{p}^{2}(\mathbb{R})$. 

Next, we classify the solutions of the Sturm-Liuville problem. We mention that a similar result has already been presented in \cite{eastham1973spectral} when $W$ and $V$ are real valued. In comparison, our classification is  from another perspective based on the spectrum and quasimomenta and is set up to the case when $W$ and $V$ are complex valued that is applied in our discussion on the multidimensional  problem.. We provide comprehensive and self-contained proofs in \ref{appdix:class-eigenfunc}.

For $k\in\mathbb{C}$, we define $[k]_{\gamma}:=\{k+\frac{2m\pi}{\gamma}, m\in\mathbb{Z}\}$. 

\begin{theorem}\label{prop:calss}
Corresponding to the eigenvalue $\lambda$, the solutions in $C_{p}^{2}(\mathbb{R})$ of \cref{eq: sl} with $W, V\in C_{p, \gamma}^{0}(\mathbb{R})$ are of one of the following forms:
\begin{enumerate}
    \item $e^{\mathrm{i}k_{0}x}p(x)$, $k_{0}\in\mathbb{C}$, $p\in C_{p, \gamma}^{2}(\mathbb{R})$;
    \item $e^{\mathrm{i}k_{0}x}(p_{1}(x)+xp_{2}(x))$, $k_{0}\in\mathbb{C}$, $p_{1}, p_{2}\in C_{p, \gamma}^{2}(\mathbb{R})$, $p_{2}\not\equiv0$;
    \item $e^{\mathrm{i}k_{1}x}p_{1}(x)+e^{\mathrm{i}k_{2}x}p_{2}(x)$, $k_{1}, k_{2}\in\mathbb{C}, [k_{1}]_{\gamma}\neq [k_{2}]_{\gamma}$, $p_{1},p_{2}\in C_{p, \gamma}^{2}(\mathbb{R})$.  
\end{enumerate}
\end{theorem}
\begin{definition}
Define the map $\mathcal{A}$ on $\mathbb{R}$ by that $\mathcal{A}(\lambda)$ is the set of the congruence classes modulo $\frac{2\pi}{\gamma}$ of quasimomenta of the solutions of \cref{eq: sl} corresponding to  $\lambda$:
\begin{equation*}
    \mathcal{A}:\mathbb{R}\rightarrow2^{\mathbb{C}/\frac{2\pi}{\gamma}\mathbb{Z}},
\end{equation*}
where $2^{\mathbb{C}/\frac{2\pi}{\gamma}\mathbb{Z}}$ is the power set of the quotient space $\mathbb{C}/\frac{2\pi}{\gamma}\mathbb{Z}$.  
\end{definition}

With the definition, we have the following conclusion.
\begin{theorem}\label{thm4.1}
For $\lambda\in\mathbb{R}$, there holds that  $|\mathcal{A}(\lambda)|\leqslant2$. Suppose $\psi\in C_{p}^{2}(\mathbb{R})$ is a solution of \cref{eq: sl} with $W, V\in C_{p, \gamma}^{0}(\mathbb{R})$. Then $\psi$ is of one of the following forms: 
\begin{itemize}
    \item if $|\mathcal{A}(\lambda)|=1$ and $\mathcal{A}(\lambda)=\{[k_{0}]_{\gamma}\}, k_{0}\in\mathbb{C}$, then  
    \begin{equation}\label{item:eig-bloch}
    \psi=e^{\mathrm{i}k_{0}x}p(x), \quad p\in C_{p, \gamma}^{2}(\mathbb{R}),
\end{equation}
or
\begin{equation}\label{item:eig-growth}
    \psi=e^{\mathrm{i}k_{0}x}(p_{1}(x)+xp_{2}(x)), p_{2}\not\equiv0, \quad p_{1},\,p_{2}\in C_{p, \gamma}^{2}(\mathbb{R}), \,p_{2}\not\equiv0;
\end{equation}
\item if $|\mathcal{A}(\lambda)|=2$ and $\mathcal{A}(\lambda)=\{[k_{1}]_{\gamma}, [k_{2}]_{\gamma}\}, k_{1}, k_{2}\in\mathbb{C}$, then 
\begin{equation}\label{item:eig-comb-bloch}
   \psi=e^{\mathrm{i}k_{1}x}p_{1}(x)+e^{\mathrm{i}k_{2}x}p_{2}(x),\quad p_{1},\,p_{2}\in C_{p, \gamma}^{2}(\mathbb{R}).
\end{equation}
\end{itemize}
\end{theorem}

\cref{thm4.1} shows that in the classical sense, any solution of the 1-D Strum-Liouville problem with $\mathcal{R}$-periodic coefficients $W, V\in C_{p}^0(\mathbb{R})$ is of one of the above three forms. As the coefficients $W, V$ and $\lambda$ vary, the solutions of \cref{eq: sl}
may be of different forms and not all be Bloch solutions. The forthcoming discussion will show that the solutions of the form \cref{item:eig-comb-bloch} indeed are the linear combination of Bloch solutions of \cref{eq: sl} corresponding to the same $\lambda$.

We set 
\begin{equation*}
\begin{array}{rcl}
    \sigma_{g}^{c}(\mathscr{L}):&=&\{\lambda\in\mathbb{R}\,|\,\mathscr{L}\psi=\lambda\psi, \psi\in C_{p}^{2}(\mathbb{R}), \psi\not\equiv0\},\\
    \sigma_{g}^{1}(\mathscr{L}):&=&\{\lambda\in\mathbb{R}\,|\,\mathscr{L}\psi=\lambda\psi, \psi\textup{ is of form \cref{item:eig-bloch} in \cref{thm4.1}}\},\\
    \sigma_{g}^{2}(\mathscr{L}):&=&\{\lambda\in\mathbb{R}\,|\,\mathscr{L}\psi=\lambda\psi,  \psi\textup{ is of form \cref{item:eig-growth} in \cref{thm4.1}}\},\\
    \sigma_{g}^{3}(\mathscr{L}):&=&\{\lambda\in\mathbb{R}\,|\,\mathscr{L}\psi=\lambda\psi, \psi\textup{ is of form \cref{item:eig-comb-bloch} in \cref{thm4.1}}\}.
\end{array}
\end{equation*}
Then, \cref{thm4.1} implies that
\begin{equation*}
    \sigma_{g}^{c}(\mathscr{L})=\sigma_{g}^{1}(\mathscr{L})\oplus\sigma_{g}^{2}(\mathscr{L})\oplus\sigma_{g}^{3}(\mathscr{L}). 
\end{equation*}
Moreover, we obtain that 
\begin{equation*}
    |\mathcal{A}(\lambda)|=\left\{\begin{array}{ll}
         1, \quad\lambda\in\sigma_{g}^{1}(\mathscr{L})\cup\sigma_{g}^{2}(\mathscr{L}),\\
         2, \quad\lambda\in\sigma_{g}^{3}(\mathscr{L}).\\
    \end{array}\right.
\end{equation*}

\begin{definition}
Given $k\in\mathbb{C}$ and 
the coefficients  $W, V\in C_{p, \gamma}^0(\mathbb{R})$, the differential operator $\mathcal{D}^{k}$ is defined as \begin{displaymath}
  \mathcal{D}^{k}:=-\frac{d^{2}}{dx^{2}}+(W-2\mathrm{i}k)\frac{d}{dx}+\mathrm{i}kW+V+k^{2}.
\end{displaymath}
\end{definition}

\begin{theorem}\label{thm4.3}
Let $\mathcal{E}(\lambda):=\{\psi|\mathscr{L}\psi=\lambda\psi, \psi\in C_{p}^{2}(\mathbb{R})\}$, then 
\begin{itemize}
    \item when $\lambda\in\sigma_{g}^{1}(\mathscr{L})$, 
    \begin{equation*}
         \mathcal{E}(\lambda)=\{e^{\mathrm{i}k_{0}x}p(x)\,|\, \mathcal{D}^{k_{0}}p=\lambda p \textup{ with } \mathcal{A}(\lambda)=\{[k_{0}]_{\gamma}\}, p\in C_{p, \gamma}^{2}(\mathbb{R})\};
    \end{equation*}
    \item when $\lambda\in\sigma_{g}^{2}(\mathscr{L})$,
    \begin{equation*}
    \begin{array}{rcl}
         \mathcal{E}(\lambda)&=&	\{e^{\mathrm{i}k_{0}x}(p_{1}(x)+xp_{2}(x))\,|\,\mathcal{D}^{k_{0}}p_{2}=\lambda p_{2},\\ &&\mathcal{D}^{k_{0}}p_{1}=\lambda p_{1}+(2\mathrm{i}k_{0}-W)p_{2}+2p_{2}', p_{1}, p_{2}\in C_{p, \gamma}^{2}(\mathbb{R})\};
    \end{array}
    \end{equation*}
    \item when $\lambda\in\sigma_{g}^{3}(\mathscr{L})$,
    \begin{equation*}
        \mathcal{E}(\lambda)=\operatorname{span}\{e^{\mathrm{i}kx}p(x)\,|\,\mathcal{D}^{k}p=\lambda p \textup{ with } [k]_\gamma\in\mathcal{A}(\lambda) \textup{ and } |\mathcal{A}(\lambda)|=2, p\in C_{p, \gamma}^{2}(\mathbb{R})\}.
    \end{equation*}
\end{itemize}
\end{theorem}

Hence, we demonstrate that for \cref{item:eig-comb-bloch}, $p_{1}$ and $p_{2}$ serve as solutions of $\mathcal{D}^{k_{1}}$ and $\mathcal{D}^{k_{2}}$, respectively. As a result, it follows that $e^{\mathrm{i}k_{1}x}p_{1}(x)$ and $e^{\mathrm{i}k_{2}x}p_{2}(x)$ are both Bloch solutions of equation \cref{eq: sl} corresponding to the same $\lambda$, which implies that when $\lambda\in\sigma_{g}^{3}(\mathscr{L})$, $\mathcal{E}(\lambda)$ is indeed spanned by Bloch solutions corresponding to the same $\lambda$ whose quasimomenta are in two different  congruence classes modulo $\frac{2\pi}{\gamma}.$

The following theorem facilitates the computation of quasimomenta for Schr{\"o}dinger equations with the proof in \ref{appdix:class-eigenfunc}.

\begin{theorem}\label{cor4.8}
Suppose that the $\mathcal{R}$-periodic coefficient $W$ is continuous. 
If $\lambda\in\sigma_{g}^{1}(\mathscr{L})\cup\sigma_{g}^{2}(\mathscr{H})$, then 
\begin{equation*}
    \mathcal{A}(\lambda)=\{[\frac{1}{2\mathrm{i}}\bar{W}]_{\gamma}\},\quad \textup{or}\quad \mathcal{A}(\lambda)=\{[\frac{1}{2\mathrm{i}}\bar{W}+\frac{\pi}{\gamma}]_{\gamma}\}
\end{equation*}
 If $\lambda\in\sigma_{g}^{3}(\mathscr{L})$, then
 \begin{equation*}
     \mathcal{A}(\lambda)=\{[k_{1}]_{\gamma}, [k_{2}]_{\gamma}\},\quad [k_{1}+k_{2}]_{\gamma}=[\frac{1}{\mathrm{i}}\bar{W}]_{\gamma},
 \end{equation*}
where $\bar{W}=\frac{1}{\gamma}\int_0^{\gamma} W(s)ds$.
\end{theorem}

We summarize \cref{prop:calss}, \cref{thm4.1}, \cref{thm4.3} and \cref{cor4.8}, as follows.

Solutions of the 1-D Strum-Liouville problem \cref{eq: sl} with  complex valued $W, V\in C_{p, \gamma}^{0}(\mathbb{R})$ are of one of the following forms:
\begin{itemize}
    \item form \cref{item:eig-bloch}, given by $\psi(x)=e^{\mathrm{i}\frac{n\pi}{\gamma}x}p(x), n\in\mathbb{Z}$, $|\mathcal{A}(\lambda)|=1$: a special class of Bloch solutions which are indeed periodic solutions i.e., $\psi(x+\gamma)=\psi(x)$ or anti-periodic solutions i.e., $\psi(x+\gamma)=-\psi(x), x\in\mathbb{R}$;
    \item form \cref{item:eig-growth}, given by $e^{\mathrm{i}\frac{n\pi}{\gamma}x}(p_{1}(x)+xp_{2}(x))$, $|\mathcal{A}(\lambda)|=1$: represents a class of special solutions  that are not Bloch solutions;
    \item  form \cref{item:eig-comb-bloch}, given by $e^{\mathrm{i}k_{1}x}p_{1}(x)+e^{\mathrm{i}k_{2}x}p_{2}(x)$ and $k_{1}+k_{2}=\frac{1}{\mathrm{i}}\bar{W}+\frac{2n\pi}{\gamma}, n\in\mathbb{Z}$, $|\mathcal{A}(\lambda)|=2$: Bloch solutions or their linear combination whose quasimomenta are in two different congruence classes modulo $\frac{2\pi}{\gamma}$, i.e., $[k_{1}]_{\gamma}\neq[k_{2}]_{\gamma}$. In fact, if $\psi_{1}(x)=e^{\mathrm{i}k_{1}x}p_{1}(x)$ is a solution,  $\psi_{2}(x)=e^{\mathrm{i}k_{2}x}p_{2}(x)$ with $k_{2}=\frac{1}{\mathrm{i}}\bar{W}-k_{1}$ and $p_{2}(x)=p_{1}((\frac{1}{\mathrm{i}k_{1}}\bar{W}-1)x)$ is also a solution corresponding to the same $\lambda$.
\end{itemize}

Note that \cref{eq: sl} becomes a 1-D Schrödinger equation as follows when $W\equiv0$. 
\begin{equation}\label{eq:sco}
    \mathscr{H}\psi:=-\frac{d^{2}\psi}{dx^{2}}+V\psi=\lambda\psi, \lambda\in\mathbb{R},
\end{equation}
where $V(x)$ is $\mathcal{R}$-periodic.
We also present several
results of 1-D periodic Schrödinger equations in \ref{appdix:class-eigenfunc}.
 
\subsection{Multidimensions}
Back to the multidimensional case \cref{eq:ddddd}, we introduce some definitions first.
\begin{definition}
$\Psi(\mathbf{r})$ is called separable in the direction of primitive vectors $\gamma_{j}(j=1, 2, \cdots, d)$ of the Bravais lattice $\mathcal{R}$ if for the invertible transformation matrix $A$,
\begin{equation*}
    A=\left(\begin{array}{cccc}
      \frac{1}{\Vert \gamma_{1}\Vert}\gamma_{1}&\frac{1}{\Vert \gamma_{2}\Vert}\gamma_{2}&\cdots&
    \frac{1}{\Vert \gamma_{d}\Vert}\gamma_{d}
    \end{array}\right)^{T}\in\mathbb{R}^{d\times d},
\end{equation*}
it holds that  $\Psi(\mathbf{r})=\tilde{\Psi}(\mathbf{\tilde{r}})=\prod_{j=1}^{d}\psi_{j}(\tilde{x}_{j})$, where for $\mathbf{r}=(x_{1},\cdots,x_{d})^{T}$, $\tilde{\mathbf{r}}=(\tilde{x_{1}},\cdots,\tilde{x_{d}})^{T}$, 
\begin{equation*}
    \tilde{\mathbf{r}}=A^{-1}\mathbf{r},
\end{equation*}
and $\psi_{j}$ is the function of one variable, $j=1,2,\cdots, d$. 
\end{definition}

\begin{definition}
For a function $f(\mathbf{r})$ defined on $\mathbb{R}^{d}$, $f(\mathbf{r})\in \bigotimes_{j=1}^{d}C_{p}^{l}(\mathbb{R})$ if $f$ is in $C_{p}^{l}(\mathbb{R})$ in the direction of $\pm\gamma_{j}, j=1, 2, \cdots, d$.
\end{definition}

 The following theorem shows a general form of separable solutions of \cref{eq:ddddd} under the assumption of $V\in \bigotimes_{j=1}^{d} C_{p}^{0}(\mathbb{R})$. 

\begin{theorem}\label{thm:dd}
    Let the $\mathcal{R}$-periodic potential $V\in \bigotimes_{j=1}^{d}C_{p}^{0}(\mathbb{R})$. If the solution $\Psi\in \bigotimes_{j=1}^{d}C_{p}^{2}(\mathbb{R})$ of \cref{eq:ddddd} satisfies that
    \begin{equation}\label{eq:ddform}
        \Psi(\mathbf{r})=\tilde{\Psi}(\mathbf{\tilde{r}})=\prod_{j=1}^{d}\psi_{j}(\tilde{x}_{j}), 
    \end{equation}
    then for $j=1,2, \cdots, d$, $\psi_{j}(t)$ is one of the following three forms: 
\begin{enumerate}
    \item $e^{\mathrm{i}k_{j}t}p_{j}(t), k_{j}\in\mathbb{C}, p_{j}\in C_{p, \Vert\gamma_{j}\Vert}^{2}(\mathbb{R})$;
    \item $e^{\mathrm{i}k_{j}t}(p_{j}^{(1)}(t)+tp_{j}^{(2)}(t)), k_{j}\in\mathbb{C}, p_{j}^{(1)}, p_{j}^{(2)}\in C_{p, \Vert\gamma_{j}\Vert}^{2}(\mathbb{R}), p_{j}^{(2)}\not\equiv0$;
    \item $e^{\mathrm{i}k_{j}^{(1)}t}p_{j}^{(1)}(t)+e^{\mathrm{i}k_{j}^{(2)}t}p_{j}^{(2)}(t), k_{j}^{(1)}, k_{j}^{(2)}\in\mathbb{C}, [k_{j}^{(1)}]_{\Vert\gamma_{j}\Vert}\neq[k_{j}^{(2)}]_{\Vert\gamma_{j}\Vert}, p_{j}^{(1)}, p_{j}^{(2)}\in C_{p, \Vert\gamma_{j}\Vert}^{2}(\mathbb{R})$.
\end{enumerate} 
\end{theorem}

\begin{proof}\label{pr:2d}
With $\Psi(\mathbf{r})=\tilde{\Psi}(\tilde{\mathbf{r}})$ and $V(\mathbf{r})=\tilde{V}(\tilde{\mathbf{r}})$, we obtain from \cref{eq:ddddd} that, 
\begin{equation}\label{eq:dtrans}
    -\frac{1}{2}\Delta_{\tilde{\mathbf{r}}}\tilde{\Psi}-\sum_{1\leqslant i<j\leqslant d}\omega_{ij}\frac{\partial^{2}\tilde{\Psi}}{\partial \tilde{x}_{i}\partial \tilde{x}_{j}}+\tilde{V}\tilde{\Psi}=\lambda\tilde{\Psi},
\end{equation}
where 
\begin{equation*}
    \Delta_{\tilde{\mathbf{r}}}\tilde{\Psi}=\sum_{j=1}^{d}\frac{\partial^{2}\tilde{\Psi}}{\partial\tilde{x_{j}}^{2}},
\end{equation*} 
$\omega_{ij}=\sum_{k=1}^{d}(a_{ik}a_{jk})$ is some constant,  $a_{ij}$ is the element in the $i$-th row and $j$-th column of $A$,   for $1\leqslant i<j\leqslant d$.

Since $\tilde{\Psi}$ is separable, we substitute $\tilde{\Psi}(\mathbf{\tilde{r}})=\prod_{j=1}^{d}\psi_{j}(\tilde{x}_{j})$ into \cref{eq:dtrans} and obtain, 
\begin{equation}\label{eq:dsep}
    \begin{array}{rcl}
         &&-\frac{1}{2}\sum_{j=1}^{d}\frac{\prod_{s=1}^{d}\psi_{s}(\tilde{x}_{s})}{\psi_{j}(\tilde{x}_{j})}\psi_{j}^{''}(\tilde{x}_{j})-\sum_{1\leqslant i<j\leqslant d}\omega_{ij}\frac{\prod_{s=1}^{d}\psi_{s}(\tilde{x}_{s})}{\psi_{i}(\tilde{x}_{i})\psi_{j}(\tilde{x}_{j})}\psi_{i}^{'}(\tilde{x}_{i})\psi_{j}^{'}(\tilde{x}_{j})\\&+&\tilde{V}(\tilde{x}_{1}, \cdots, \tilde{x}_{d})\prod_{j=1}^{d}\psi_{j}(\tilde{x}_{j})=\lambda\prod_{j=1}^{d}\psi_{j}(\tilde{x}_{j}).
    \end{array}
\end{equation}

For fixed $\tilde{x}_{j}=\tilde{x}_{j}^{(0)}, \forall j=2, \cdots, d$, \cref{eq:dsep} is a Sturm-Liouville problem with respect to $\psi_{1}$, where reads as, 

\begin{equation}
    \begin{array}{rcl}
         &&-\frac{1}{2}\psi_{1}^{''}-\sum_{j=2}^{d}\omega_{1j}\frac{\psi_{j}'(\tilde{x}_{j}^{(0)})}{\psi_{j}(\tilde{x}_{j}^{(0)})}\psi_{1}^{'}
    +\big[\tilde{V}(\cdot, \tilde{x}_{2}^{(0)}, \ldots, \tilde{x}_{d}^{(0)})\\&-&\frac{1}{2}\sum_{j=2}^{d}\frac{\psi_{j}^{''}(\tilde{x}_{j}^{(0)})}{\psi_{j}(\tilde{x}_{j}^{(0)})}-\sum_{2\leqslant i<j\leqslant d}\omega_{ij}\frac{\psi_{i}^{'}(\tilde{x}_{i}^{(0)})\psi_{j}^{'}(\tilde{x}_{j}^{(0)})}{\psi_{i}(\tilde{x}_{i}^{(0)})\psi_{j}(\tilde{x}_{j}^{(0)})}\big]\psi_{1}=\lambda\psi_{1}.
    \end{array}
\end{equation}

 Since $\sum_{j=2}^{d}\omega_{1j}\frac{\psi_{j}'(\tilde{x}_{j}^{(0)})}{\psi_{j}(\tilde{x}_{j}^{(0)})}$ is  constant and 
 \begin{equation*}
     \tilde{V}(\cdot, \tilde{x}_{2}^{(0)}, \cdots, \tilde{x}_{d}^{(0)})-\frac{1}{2}\sum_{j=2}^{d}\frac{\psi_{j}^{''}(\tilde{x}_{j}^{(0)})}{\psi_{j}(\tilde{x}_{j}^{(0)})}-\sum_{2\leqslant i<j\leqslant d}\omega_{ij}\frac{\psi_{i}^{'}(\tilde{x}_{i}^{(0)})\psi_{j}^{'}(\tilde{x}_{j}^{(0)})}{\psi_{i}(\tilde{x}_{i}^{(0)})\psi_{j}(\tilde{x}_{j}^{(0)})}
 \end{equation*}
 is periodic,  we obtain from \cref{thm4.1} and \cref{thm4.3} that $\psi_{1}$ is of one of the following three forms,
\begin{enumerate}
    \item $e^{\mathrm{i}k_{v_{0}}\tilde{x}_{1}}p_{v_{0}}(\tilde{x}_{1})$,
    \item $e^{\mathrm{i}k_{v_{0}}\tilde{x}_{1}}(p_{v_{0}}^{(1)}(\tilde{x}_{1})+\tilde{x}_{1}p_{v_{0}}^{(2)}(\tilde{x}_{1}))$,
    \item $e^{\mathrm{i}k_{v_{0}}^{(1)}\tilde{x}_{1}}p_{v_{0}}^{(1)}(\tilde{x}_{1})+e^{\mathrm{i}k_{v_{0}}^{(2)}\tilde{x}_{1}}p_{v_{0}}^{(2)}(\tilde{x}_{1})$,
\end{enumerate}
where $k_{v_{0}}, k_{v_{0}}^{(1)}, k_{v_{0}}^{(2)}\in\mathbb{C}$ and $p_{v_{0}}, p_{v_{0}}^{(1)}, p_{v_{0}}^{(2)}\in C_{p}^{2}(\mathbb{R})$ are periodic and may be related to $v_{0}$. For convenience, we let $v_{s}=(\tilde{x}_{2}^{(s)}, \cdots, \tilde{x}_{d}^{(s)})^{T}$.

Next, we show that we can choose the coefficients $k_{v_{0}}, k_{v_{0}}^{(1)}, k_{v_{0}}^{(2)}$ and $p_{v_{0}}, p_{v_{0}}^{(1)}, p_{v_{0}}^{(2)}$ that are independent from $v_{0}$.

Note that $\psi_{1}(\tilde{x}_{1})$ depends only on $\tilde{x}_{1}$. Hence, for the case that $\psi_{1}$ is the Bloch solution (or their linear combination), i.e.,  $\psi_{1}(\tilde{x}_{1})=e^{\mathrm{i}k_{v_{0}}\tilde{x}_{1}}p_{v_{0}}(\tilde{x}_{1})$ , we obtain that, for $v_{1}\neq v_{2}$, 
\begin{equation}\label{eq:depen}
\psi_{1}(\tilde{x}_{1})=e^{\mathrm{i}k_{v_{1}}\tilde{x}_{1}}p_{v_{1}}(\tilde{x}_{1})=e^{\mathrm{i}k_{v_{2}}\tilde{x}_{1}}p_{v_{2}}(\tilde{x}_{1})
\end{equation}

Then, the periodicity yields
\begin{equation*}
    e^{\mathrm{i}(k_{v_{1}}-k_{v_{2}})\tilde{x}_{1}}=\frac{p_{v_{2}}(\tilde{x}_{1})}{p_{v_{1}}(\tilde{x}_{1})}=\frac{p_{v_{2}}(\tilde{x}_{1}+\Vert\gamma_{1}\Vert)}{p_{v_{1}}(\tilde{x}_{1}+\Vert\gamma_{1}\Vert)}=e^{\mathrm{i}(k_{v_{1}}-k_{v_{2}})(\tilde{x}_{1}+\Vert\gamma_{1}\Vert}), \forall \tilde{x}_{1}\in\mathbb{R},
\end{equation*}
and 
\begin{equation}
    e^{\mathrm{i}(k_{v_{1}}-k_{v_{2}})\Vert\gamma_{1}\Vert}=1,
\end{equation}
which implies that for $v\in\mathbb{R}^{d-1}$, it holds that $k_{v}=k_{v=0}+\frac{2n_{v}\pi}{\Vert\gamma_{1}\Vert}$, where $n_{v}$ is some integer that depends on $v$.

Hence, we get 
\begin{equation}\label{eq:dif}
    \psi_{1}(\tilde{x}_{1})=e^{\mathrm{i}k_{1}\tilde{x}_{1}}e^{\mathrm{i}\frac{2n_{v}\pi}{\Vert\gamma_{1}\Vert}\tilde{x}_{1}}p_{v}(\tilde{x}_{1}).
\end{equation}

Since $\psi_{1}(\tilde{x}_{1})$ and $e^{\mathrm{i}k_{1}\tilde{x}_{1}}$ depends only on $\tilde{x}_{1}$, we obtain from \cref{eq:dif} that $e^{\mathrm{i}\frac{2n_{v}\pi}{\Vert\gamma_{1}\Vert}\tilde{x}_{1}}p_{v}(\tilde{x}_{1})$ also depend only on $\tilde{x}_{1}$. $p_{1}(\tilde{x}_{1})=e^{\mathrm{i}\frac{2n_{v}\pi}{\Vert\gamma_{1}\Vert}\tilde{x}_{1}}p_{v}(\tilde{x}_{1})$ has the period $\Vert\gamma_{1}\Vert$. And $\psi_{1}(\tilde{x}_{1})=e^{\mathrm{i}k_{1}\tilde{x}_{1}}p_{1}(\tilde{x}_{1})$ is the Bloch solution that is independent of the variables $\tilde{x}_{j}^{(0)}, j=2,\cdots, d$.

For the case that $\psi_{1}$ is not the Bloch solution (or their linear combination), i.e., $\psi_{1}(\tilde{x}_{1})=e^{\mathrm{i}k_{v_{0}}\tilde{x}_{1}}(p_{v_{0}}^{(1)}(\tilde{x}_{1})+\tilde{x}_{1}p_{v_{0}}^{(2)}(\tilde{x}_{1}))$,  for $v_{1}\neq v_{2}$,
\begin{equation}\label{eq:special}
    e^{\mathrm{i}k_{v_{1}}\tilde{x}_{1}}(p_{v_{1}}^{(1)}(\tilde{x}_{1})+\tilde{x}_{1}p_{v_{1}}^{(2)}(\tilde{x}_{1}))=e^{\mathrm{i}k_{v_{2}}\tilde{x}_{1}}(p_{v_{2}}^{(1)}(\tilde{x}_{1})+\tilde{x}_{1}p_{v_{2}}^{(2)}(\tilde{x}_{1})), 
\end{equation}
then we can claim that $\operatorname{Im}(k_{v_{1}}-k_{v_{2}})=0$. Let us prove it by contradiction.

If $\operatorname{Im}(k_{v_{1}}-k_{v_{2}})\neq0$,  without loss of generality, we suppose $\operatorname{Im}(k_{v_{1}}-k_{v_{2}})>0$. Since $p_{v_{1}}^{(1)}, p_{v_{1}}^{(2)}, p_{v_{2}}^{(1)}, p_{v_{2}}^{(2)}$ are periodic, we obtain from \cref{eq:special} that 
\begin{equation*}
    e^{\mathrm{i}(k_{v_{1}}-k_{v_{2}})(\tilde{x}_{1}+n\Vert\gamma_{1}\Vert)}=\frac{p_{v_{2}}^{(1)}(\tilde{x}_{1})+(\tilde{x}_{1}+n\Vert\gamma_{1}\Vert)p_{v_{2}}^{(2)}(\tilde{x}_{1})}{p_{v_{1}}^{(1)}(\tilde{x}_{1})+(\tilde{x}_{1}+n\Vert\gamma_{1}\Vert)p_{v_{1}}^{(2)}(\tilde{x}_{1})}, \quad \forall n\in\mathbb{Z}.
\end{equation*}

For ${\tilde x}_1^{(0)}$, if $p_{v_{1}}^{(2)}(\tilde{x}_{1}^{(0)})\neq0$, then
\begin{equation}\label{eq:seprai}
    \begin{array}{rcl}
         &&e^{\mathrm{i}\operatorname{Re}(k_{v_{1}}-k_{v_{2}})(\tilde{x}_{1}^{(0)}+n\Vert\gamma_{1}\Vert)}e^{-\operatorname{Im}(k_{v_{1}}-k_{v_{2}})(\tilde{x}_{1}^{(0)}+n\Vert\gamma_{1}\Vert)}\\&=&\frac{p_{v_{2}}^{(1)}(\tilde{x}_{1}^{(0)})+(\tilde{x}_{1}^{(0)}+n\Vert\gamma_{1}\Vert)p_{v_{2}}^{(2)}(\tilde{x}_{1}^{(0)})}{p_{v_{1}}^{(1)}(\tilde{x}_{1}^{(0)})+(\tilde{x}_{1}^{(0)}+n\Vert\gamma_{1}\Vert)p_{v_{1}}^{(2)}(\tilde{x}_{1}^{(0)})}.
    \end{array}
\end{equation}

As $n\rightarrow-\infty$, the left side of \cref{eq:seprai} tends to infinity while the right side tends to $\frac{p_{v_{2}}^{(2)}(\tilde{x}_{1}^{(0)})}{p_{v_{1}}^{(2)}(\tilde{x}_{1}^{(0)})}$, which is a constant. Contradiction! 

Hence, $\operatorname{Im}(k_{v_{1}}-k_{v_{2}})=0$. Then, we derive from \cref{eq:special} that for $\tilde{x}_{1}\in\mathbb{R}$, 
\begin{equation}\label{eq:munus}
    e^{\mathrm{i}k_{v_{2}}\tilde{x}_{1}}p_{v_{2}}^{(1)}(\tilde{x}_{1})-e^{\mathrm{i}k_{v_{1}}\tilde{x}_{1}}p_{v_{1}}^{(1)}(\tilde{x}_{1})=(e^{\mathrm{i}k_{v_{1}}\tilde{x}_{1}}p_{v_{1}}^{(2)}(\tilde{x}_{1})-e^{\mathrm{i}k_{v_{2}}\tilde{x}_{1}}p_{v_{2}}^{(2)}(\tilde{x}_{1}))\tilde{x}_{1}.
\end{equation}

Next we will show 
\begin{equation}
    e^{\mathrm{i}k_{v_{1}}\tilde{x}_{1}}p_{v_{1}}^{(2)}(\tilde{x}_{1})=e^{\mathrm{i}k_{v_{2}}\tilde{x}_{1}}p_{v_{2}}^{(2)}(\tilde{x}_{1}),\quad \forall \tilde{x}_1\in\mathbb{R}. 
\end{equation}

Let us prove it by contradiction. Suppose there exists $\tilde{x}_{1}^{(0)}\in\mathbb{R}$ such that  $e^{\mathrm{i}k_{v_{1}}\tilde{x}_{1}^{(0)}}p_{v_{1}}^{(2)}(\tilde{x}_{1}^{(0)})\neq e^{\mathrm{i}k_{v_{2}}\tilde{x}_{1}^{(0)}}p_{v_{2}}^{(2)}(\tilde{x}_{1}^{(0)})$. We set $k_{v_{1}}\not\equiv k_{v_{2}} (\operatorname{mod}\frac{2\pi}{\Vert\gamma_{1}\Vert})$ and $p_{v_{2}}^{(2)}(\tilde{x}_{1}^{(0)})\neq0$ otherwise it is trivial. Then we obtain from the periodicity that 
\begin{equation}\label{eq:minmm}
    \varepsilon\coloneqq|(e^{\mathrm{i}(k_{v_{1}}-k_{v_{2}})\Vert\gamma_{1}\Vert}-1)p^{(2)}_{v_{2}}(\tilde{x}_{1}^{(0)})|>0.
\end{equation}

We have from \cref{eq:munus} that for $n\in\mathbb{Z}$, 
\begin{equation}\label{eq:munusper}
\begin{array}{rcl}
     &&e^{\mathrm{i}(k_{v_{2}}-k_{v_{1}})(\tilde{x}_{1}^{(0)}+n\Vert\gamma_{1}\Vert)}p_{v_{2}}^{(1)}(\tilde{x}_{1}^{(0)})-p_{v_{1}}^{(1)}(\tilde{x}_{1}^{(0)})\\&=&\left(e^{\mathrm{i}(k_{v_{2}}-k_{v_{1}})(\tilde{x}_{1}^{(0)}+n\Vert\gamma_{1}\Vert)}p_{v_{2}}^{(2)}(\tilde{x}_{1}^{(0)})-p_{v_{1}}^{(2)}(\tilde{x}_{1}^{(0)})\right)(\tilde{x}_{1}^{(0)}+n\Vert\gamma_{1}\Vert).
\end{array}
\end{equation}

If there exists $N_{0}>0$, such that 
\begin{equation*}
    |e^{\mathrm{i}(k_{v_{2}}-k_{v_{1}})(\tilde{x}_{1}^{(0)}+n\Vert\gamma_{1}\Vert)}p_{v_{2}}^{(2)}(\tilde{x}_{1}^{(0)})-p_{v_{1}}^{(2)}(\tilde{x}_{1}^{(0)})|<\frac{\varepsilon}{2}, \quad\forall n>N_{0},
\end{equation*}
then $k_{v_{1}}-k_{v_{2}}\in\mathbb{R}$ leads to 
\begin{equation*}
\begin{array}{rcl}
     \varepsilon&>&|e^{\mathrm{i}(k_{v_{2}}-k_{v_{1}})(\tilde{x}_{1}^{(0)}+(n+1)\Vert\gamma_{1}\Vert)}p_{v_{2}}^{(2)}(\tilde{x}_{1}^{(0)})-e^{\mathrm{i}(k_{v_{2}}-k_{v_{1}})(\tilde{x}_{1}^{(0)}+n\Vert\gamma_{1}\Vert)}p_{v_{2}}^{(2)}(\tilde{x}_{1}^{(0)})|\\&=&|e^{\mathrm{i}(k_{v_{2}}-k_{v_{1}})(\tilde{x}_{1}^{(0)}+n\Vert\gamma_{1}\Vert)}(e^{\mathrm{i}(k_{v_{1}}-k_{v_{2}})\Vert\gamma_{1}\Vert}-1)p^{(2)}_{v_{2}}(\tilde{x}_{1}^{(0)})|\\&=&|(e^{\mathrm{i}(k_{v_{1}}-k_{v_{2}})\Vert\gamma_{1}\Vert}-1)p^{(2)}_{v_{2}}(\tilde{x}_{1}^{(0)})|,
\end{array}
\end{equation*}
which is contrary to \cref{eq:minmm}.

Hence, there exists a sequence $\{n_{t}\}_{t\in\mathbb{Z}}\subset\mathbb{Z}$ such that
\begin{equation*}
   |e^{\mathrm{i}(k_{v_{2}}-k_{v_{1}})(\tilde{x}_{1}^{(0)}+n_{t}\Vert\gamma_{1}\Vert)}p_{v_{2}}^{(2)}(\tilde{x}_{1}^{(0)})-p_{v_{1}}^{(2)}(\tilde{x}_{1}^{(0)})|\geqslant\frac{\varepsilon}{2}. 
\end{equation*}

Note that $k_{v_{2}}-k_{v_{1}}\in\mathbb{R}$ and $p_{v_{1}}^{(1)}, p_{v_{1}}^{(2)}, p_{v_{2}}^{(1)}, p_{v_{2}}^{(2)}\in C_{p, \Vert\gamma_{1}\Vert}^{2}(\mathbb{R})$. We set $M>0$ as the upper bound of $p_{v_{i}}^{(j)}, i, j=1, 2$ and arrive at that for $j=1, 2$, 
\begin{equation*}
\begin{array}{rcl}
    &&|e^{\mathrm{i}(k_{v_{2}}-k_{v_{1}})(\tilde{x}_{1}^{(0)}+n\Vert\gamma_{1}\Vert)}p_{v_{2}}^{(j)}(\tilde{x}_{1}^{(0)})-p_{v_{1}}^{(j)}(\tilde{x}_{1}^{(0)})|\\&\leqslant&|e^{\mathrm{i}(k_{v_{2}}-k_{v_{1}})(\tilde{x}_{1}^{(0)}+n\Vert\gamma_{1}\Vert)}||p_{v_{2}}^{(j)}(\tilde{x}_{1}^{(0)})|+|p_{v_{1}}^{(j)}(\tilde{x}_{1}^{(0)})|\\&=&|p_{v_{2}}^{(j)}(\tilde{x}_{1}^{(0)})|+|p_{v_{1}}^{(j)}(\tilde{x}_{1}^{(0)})|\\&\leqslant&2M.
\end{array}
\end{equation*}
 Then  we set $n=n_{t}$ and have
 \begin{equation*}
 \begin{array}{rcl}
      &&|(e^{\mathrm{i}(k_{v_{2}}-k_{v_{1}})(\tilde{x}_{1}^{(0)}+n_{t}\Vert\gamma_{1}\Vert)}p_{v_{2}}^{(2)}(\tilde{x}_{1}^{(0)})-p_{v_{1}}^{(2)}(\tilde{x}_{1}^{(0)}))(\tilde{x}_{1}^{(0)}+n_{t}\Vert\gamma_{1}\Vert)|\\&=&|(e^{\mathrm{i}(k_{v_{2}}-k_{v_{1}})(\tilde{x}_{1}^{(0)}+n_{t}\Vert\gamma_{1}\Vert)}p_{v_{2}}^{(2)}(\tilde{x}_{1}^{(0)})-p_{v_{1}}^{(2)}(\tilde{x}_{1}^{(0)}))||\tilde{x}_{1}^{(0)}+n_{t}\Vert\gamma_{1}\Vert|\\&\geqslant&\frac{\varepsilon}{2}|\tilde{x}_{1}^{(0)}+n_{t}\Vert\gamma_{1}\Vert|.
 \end{array}
 \end{equation*}
 
 As $t\rightarrow\infty$, the right side of \cref{eq:munusper} tends towards infinity while the left side is bounded. It is a contradiction. 

Hence, we have
\begin{equation}
    e^{\mathrm{i}k_{v_{1}}\tilde{x}_{1}}p_{v_{1}}^{(2)}(\tilde{x}_{1})=e^{\mathrm{i}k_{v_{2}}\tilde{x}_{1}}p_{v_{2}}^{(2)}(\tilde{x}_{1}), \quad\forall \tilde{x}_1\in\mathbb{R}
\end{equation}
for any $v_1\ne v_2$ and in accordance with \cref{eq:munus}, 
\begin{equation}
    e^{\mathrm{i}k_{v_{2}}\tilde{x}_{1}}p_{v_{2}}^{(1)}(\tilde{x}_{1})=e^{\mathrm{i}k_{v_{1}}\tilde{x}_{1}}p_{v_{1}}^{(1)}(\tilde{x}_{1}). 
\end{equation}

Similarly, we can choose $k_{1}\in\mathbb{C}$ and $p_{1}^{(1)}, p_{1}^{(2)}\in C_{p, \Vert\gamma_{1}\Vert}^{2}(\mathbb{R})$, that are independent of the variables $\tilde{x}_{j}^{(0)}, j=2,\cdots, d$, and we arrive at that $\psi_{1}(t)$ is of one of the following three forms,
\begin{enumerate}
    \item $e^{\mathrm{i}k_{1}t}p_{1}(t)$;
    \item $e^{\mathrm{i}k_{1}t}(p_{1}^{(1)}(t)+tp_{1}^{(2)}(t)), p_{1}^{(2)}\not\equiv0$;
    \item $e^{\mathrm{i}k_{1}^{(1)}t}p_{1}^{(1)}(t)+e^{\mathrm{i}k_{1}^{(2)}t}p_{1}^{(2)}(t)$, 
\end{enumerate} 

Also, we can obtain that  $\psi_{j}(t)$ is of one of the following three forms,
\begin{enumerate}
   \item $e^{\mathrm{i}k_{j}t}p_{j}(t), k_{j}\in\mathbb{C}, p_{j}\in C_{p, \Vert\gamma_{j}\Vert}^{2}(\mathbb{R})$;
    \item $e^{\mathrm{i}k_{j}t}(p_{j}^{(1)}(t)+tp_{j}^{(2)}(t)), k_{j}\in\mathbb{C}, p_{j}^{(1)}, p_{j}^{(2)}\in C_{p, \Vert\gamma_{j}\Vert}^{2}(\mathbb{R}), p_{j}^{(2)}\not\equiv0$;
    \item $e^{\mathrm{i}k_{j}^{(1)}t}p_{j}^{(1)}(t)+e^{\mathrm{i}k_{j}^{(2)}t}p_{j}^{(2)}(t), k_{j}^{(1)}, k_{j}^{(2)}\in\mathbb{C}, [k_{j}^{(1)}]_{\Vert\gamma_{j}\Vert}\neq[k_{j}^{(2)}]_{\Vert\gamma_{j}\Vert}, p_{j}^{(1)}, p_{j}^{(2)}\in C_{p, \Vert\gamma_{j}\Vert}^{2}(\mathbb{R})$,
\end{enumerate} 
which implies \cref{eq:ddform}.
\end{proof}

Similar to the one dimensional cases, \cref{thm:dd} indeed presents a classification of the separable solutions of \cref{eq:ddddd} as follows:
\begin{itemize}
    \item the Bloch solution, given by $\Psi(\mathbf{r})=\tilde{\Psi}(\tilde{\mathbf{r}})=e^{\mathrm{i}\mathbf{k}\cdot\tilde{\mathbf{r}}}p(\tilde{\mathbf{r}})$, $\mathbf{k}=(k_{1}, k_{2}, \cdots, k_{d})^{T}$, or the linear combination whose $j$-th components of the quasimomenta are in at most two congruence
classes  module $\frac{2\pi}{\Vert\gamma_{j}\Vert}, j=1, 2, \cdots, d$; 
\item  and the special solution that is not a Bloch solution, given by $\Psi(\mathbf{r})=\prod_{j=1}^{d}\psi_{j}(\tilde{x}_{j})$, where some of $\psi_{j}, j=1, 2, \cdots, d$ are of the form \cref{item:eig-growth}.
\end{itemize}

Now we are able to address the relationship between bounded solutions and Bloch solutions in multidimensions.
 
\begin{theorem}\label{bloch thm}
The bounded separable solutions of \cref{eq:ddddd} with the periodic potential $V\in\bigotimes_{j=1}^{d}C_{p}^{0}(\mathbb{R})$ are either Bloch solutions or their linear combinations  with real quasimomenta.
\end{theorem}
\begin{proof}
Consider the bounded separable solution $\Psi$ of \cref{eq:ddddd}. In accordance with \cref{thm:dd}, we obtain that $ \Psi(\mathbf{r})=\tilde{\Psi}(\mathbf{\tilde{r}})=\prod_{j=1}^{d}\psi_{j}(\tilde{x}_{j})$, where $\psi_{j}$ is the function of one variable that of one of the following three forms presented above.

Since $\Psi$ is bounded, we have that $\psi_{j}$ is bounded, for $j\geqslant1$. The arguments in the proof of \cref{pro} shows that $\psi_{j} (j\geqslant1)$ is either a Bloch solution or their combination with real quasimomenta, which completes the proof.
\end{proof}

Next  we present some properties of Bloch solutions and quasimomenta based on above results. For convenience, we set the Bravais lattice $\mathcal{R}=\mathbb{Z}^{d}$.
Then it is clear that the (first) Brillouin zone $\mathcal{B}=[-\pi, \pi]^{d}$.

 If we restrict all the quasimomenta into $\mathcal{B}$, in other words, the the quotient space $\mathbb{R}^{d}/\mathbb{Z}^{d}$, then for the separable solution $\Psi(\mathbf{r})=\prod_{j=1}^{d}\psi_{j}(x_{j})$,  \cref{col:two} and \cref{thm:dd} tell that $\psi_{j}$ is the linear combination of one or two Bloch solutions for $j=1, 2, \cdots, d$. 
 \begin{corollary}
 Under the same assumption in \cref{bloch thm}, the bounded separable solutions of \cref{eq:ddddd} are linear combinations of $2^{s}$ Bloch solutions corresponding to one $\lambda$, $s=1, 2, \cdots, d$,  and the quasimomenta of these  $2^{s}$ Bloch solutions are in the symmetric distribution about the coordinate axises of $\mathbb{R}^{d}$.  
 \end{corollary}

We sketch an example to illustrate separable solutions of \cref{eq:ddddd}.

\begin{example}\label{eq:examp3}
Let the Bravais lattice $\mathcal{R}=\mathbb{Z}^{d}$, consider the following Hartree-type equation, 
\begin{equation}\label{eq:exammm}
    -\Delta\Psi+(\sum_{j=1}^{d}V_{j}(x_{j}))\Psi=E\Psi, 
\end{equation}
where the potential $V_{j}\in C_{p}^{0}(\mathbb{R})$ is in the period $1$, for $j=1, 2, \cdots, d$.

Then the separable solution $\Psi(x_{1}, \cdots, x_{d})=\prod_{j=1}^{d}\psi_{j}(x_{j})$ is the solution of \cref{eq:exammm} corresponding to $E=\sum_{j=1}^{d}E_{j}$ , where $\psi_{j}$ is the solution of the following 1-D Schrödinger equations, 
\begin{equation*}
    -\frac{d^{2}}{dx^{2}}\psi_{j}+V_{j}\psi_{j}=E_{j}\psi_{j}.
\end{equation*}

In accordance with the conclusions in one dimension, for $j=1, 2, \cdots, d$, $\psi_{j}$ is of the three forms mentioned in \cref{thm4.1}.
\end{example}

\section{Concluding remarks}
In this paper, we have shown a mathematical aspect on the Bloch theorem. By studying a more general class of 1-D Sturm-Liouville problems with periodic coefficients, 
we provide a classification of solutions of periodic Schrodinger equations when they are separable solutions whose variables can be separated directly or with some linear mapping in multidimensions. In addition, we have shown some properties of solutions and the quasimomenta derived from the classification. 

We have addressed the relationship between bounded solutions and Bloch solutions and have proved that the bounded separable solutions are either Bloch solutions or their linear combinations. We believe that there exist some relationships between the separable solutions and general solutions of Schrödinger operators with periodic potentials. For \cref{eq:ddddd} with a constant potential in $2$ dimensions, for instance, we see that any solution  can be represented as a sum or integral of these separable solutions (see, e.g.,  Chapter I of \cite{miller1977symmetry} ). We conjecture that any bounded solution of Schrödinger operators with periodic potentials are linear combinations of Bloch solutions or their limitations,  which is indeed supported by the plane wave approach based electronic structure calculations

\appendix
\section{Details of one dimensional cases}\label{appdix:class-eigenfunc}

We first introduce some definitions and notations. 
\begin{definition}
A function $f(t)$ on $(-\infty, \infty)$ is said to be piecewise-continuous if it has only finitely many discontinuity points on any finite interval.
        \end{definition}

        \begin{definition}\label{def:int}
        A piecewise-continuous function $f$ is said to be integrable if, at any discontinuity point $t^{*}$, the limits
        \begin{equation*}
        \begin{array}{rcl}
            \lim_{\epsilon\rightarrow0^{+}}\int_{t_{1}}^{t^{*}-\epsilon}|f(t)|dt&=\int_{t_{1}}^{t^{*}}|f(t)|dt,\\
           \lim_{\epsilon\rightarrow0^{+}}\int_{t^{*}+\epsilon}^{t_{2}}|f(t)|dt&=\int_{t^{*}}^{t_{2}}|f(t)|dt
        \end{array}
        \end{equation*}
        exist, where $t_{1}, t_{2}$ are so chosen that the intervals $(t_{1}, t^{*})$, $(t^{*}, t_{2})$ contain no other discontinuity points.
        \end{definition}

\begin{definition}\label{def:pc}
      Let $C_{p}^{0}(I)$ be the class of the functions defined on the interval $I$ that are integrable and piecewise-continuous. Recursively, denote $C_{p}^{l}(I)$ the class of the functions that are differentiable on $I$ with its derivative in $C_{p}^{l-1}(I), l=1, 2, 3, \ldots$ Let $C_{p, \gamma}^{l}(I)$ be the class of the functions in $C_{p}^{l}(I)$ and in the period of $\gamma$, $\gamma>0, l=1, 2, 3, \ldots$
\end{definition}

We sketch the Floquet-Lyapunov theory for ordinary differential equations below (see,  e.g., \cite{kuchment2016overview} and  Chapter II of  \cite{Yakubovich1975linear}). Consider a linear system  
\begin{equation}\label{eq3.1}
    \frac{dx}{dt}=A(t)x,\quad x\in\mathbb{C}^{n}, t\in\mathbb{R},
\end{equation}
where the complex-valued $n\times n$ matrix function $A(t)$ has a period $T\,(T>0)$, i.e. $A(t+T)=A(t), \forall t>0$.
\begin{definition}
The $n\times n$ matrix $X(t)$ is called the fundamental matrix if its $j$-th column $x_{j}(t)$ is the solution of \cref{eq3.1} with the initial condition $x_{j}(0)=e_{j}$, where $e_{j}$ is the $j$-th column of the identity matrix, for $j=1,2,\ldots,n$.
\end{definition}

The existence and uniqueness of the fundamental matrix are proved in  $(C_{p}^{1}(\mathbb{R}))^{n\times n}$ (see Chapter II of \cite{Yakubovich1975linear} for more details).  With the fundamental matrix $X(t)$, any solution $x(t)$ of \cref{eq3.1} can be written as
\begin{equation}\label{eq3.2}
    x(t)=X(t)x(0).
\end{equation}

Moreover, the $n\times n$ matrix $M:=X(T)$ is called the monodromy matrix of \cref{eq3.1}.

The following Floquet theorem, whcih can be founded in \cite{Yakubovich1975linear}, will be applied in our discussion
\begin{lemma}\label{thm3.3}If $A(t)\in (C_{p}^{0}(\mathbb{R}))^{n\times n}$, then there exists a $T$-periodic matrix function $P(t)\in (C_{p}^{1}(\mathbb{R}))^{n\times n}$ and a constant matrix $C$, such that the fundamental matrix $X(t)$ of \cref{eq3.1} is of the form
\begin{equation}\label{eq3.3}
    X(t)=P(t)e^{Ct}.
\end{equation}
\end{lemma}

Note that there holds
\begin{equation}\label{eq3.4}
    P(0)=I,
\end{equation}
where $I$ is the identity matrix. Hence, we obtain from the periodicity of $P$ that
\begin{equation}\label{eq3.5}
    M=X(T)=e^{CT}.
\end{equation}

\begin{proof}[Proof of \cref{prop:calss}]
Let $\displaystyle\phi=\frac{d\psi}{dx}$ and $U=(\psi, \phi)^{T}$. It follows from \cref{eq: sl} that
\begin{equation}\label{eq4.2}
    \frac{dU}{dx}=\left(\begin{array}{cc}
      0&1\\V-\lambda&W
    \end{array}\right)U
\end{equation}
 
Note that $\left(\begin{array}{cc}
      0&1\\V-\lambda&W
    \end{array}\right)$ is an integrable, piecewise-continuous and $\mathcal{R}$-periodic matrix function. We obtain from \cref{thm3.3} that there eixsts $\mathcal{R}$-periodic matrix function $P(x)\in (C_{p}^{1}(\mathbb{R}))^{2\times2}$ and a constant matrix $C$ such that the the fundamental matrix $X(x)$ satisfies
\begin{equation}\label{eq4.3}
    X(x)=P(x)e^{Cx}.
\end{equation}

Moreover, if $M$ denotes the monodromy matrix of \cref{eq4.2}, then \cref{eq3.4} and \cref{eq3.5} imply,
\begin{equation}\label{eq4.4}
    M=e^{C\gamma}.
\end{equation}

Since $C\in\mathbb{C}^{2\times 2}$, there are three possible Jordan normal forms of $C$ as follows, 
\begin{equation}\label{eq4.5}
    J_{1}=\left(\begin{array}{cc}
      z_{0}&0\\0&z_{0}
    \end{array}\right), 
    J_{2}=\left(\begin{array}{cc}
      z_{0}&1\\0&z_{0}
    \end{array}\right),
    J_{3}=\left(\begin{array}{cc}
      z_{1}&0\\0&z_{2}
    \end{array}\right),
\end{equation}
where $z_{0}, z_{1}, z_{2}\in\mathbb{C}$. And there exist invertible matrices $T_{j}$ satisfying $T_{j}^{-1}CT_{j}=J_{j}$, for $j=1,2,3$.

Let $Q_{j}(x)=T_{j}^{-1}P(x)T_{j}$, we obtain that
\begin{equation}
\begin{array}{rcl}
     X(x)&=&T_{j}(T_{j}^{-1}P(x)T_{j}T_{j}^{-1}e^{Cx}T_{j})T_{j}^{-1}\\
  &=&T_{j}(Q_j(x)e^{T_{j}^{-1}CT_{j}x})T_{j}^{-1}\\
  &=&T_{j}(Q_j(x)e^{J_{j}x})T_{j}^{-1}\label{eq4.6},
\end{array}
\end{equation}
for $j=1, 2, 3$. Moreover, we have that
  \begin{equation}\label{eq4.7}
      e^{J_{1}x}=\left(\begin{array}{cc}
      e^{z_{0}x}&0\\0&e^{z_{0}x}
    \end{array}\right), e^{J_{2}x}=\left(\begin{array}{cc}
      e^{z_{0}x}&xe^{z_{0}x}\\0&e^{z_{0}x}
    \end{array}\right),
      e^{J_{3}x}=\left(\begin{array}{cc}
      e^{z_{1}x}&0\\0&e^{z_{2}x}
    \end{array}\right).
  \end{equation}
  
  Since $Q_{j}\in (C_{p}^1(\mathbb{R}))^{2\times 2}$ is $\mathcal{R}$-periodic, 
  %continuous and $\mathcal{R}$-periodic with integrable piecewise-continuous derivatives for $j=1,2,3$, 
  we get from \cref{eq3.2}, \cref{eq4.6}, and \cref{eq4.7} that there exist $\mathcal{R}$-periodic functions $p, p_{1}, p_{2}\in C_{p}^{1}(\mathbb{R})$ such that $\psi(x)$ is of one of the following forms:

\begin{enumerate}
    \item $e^{z_{0}x}p(x)$,
    \item $e^{z_{0}x}(p_{1}(x)+xp_{2}(x))$,
    \item $e^{z_{1}x}p_{1}(x)+e^{z_{2}x}p_{2}(x)$.
\end{enumerate}

If there exists $m\in\mathbb{Z}$ such that $z_{1}=z_{2}+\mathrm{i}\frac{2m\pi}{\gamma}$  in the case of (iii),  then
\begin{equation*}
\begin{array}{rcl}
     \psi(x)&=&e^{z_{1}x}p_{1}(x)+e^{z_{2}x}p_{2}(x)\\
    &=&e^{z_{2}x}(e^{\mathrm{i}\frac{2m\pi}{\gamma}x}p_{1}(x)+p_{2}(x)).
\end{array}
\end{equation*}

Note that $e^{\mathrm{i}\frac{2m\pi}{\gamma}x}p_{1}(x)+p_{2}(x)$ is $\mathcal{R}$-periodic. Hence, $\psi(x)$ is of the form (i), for which we set $z_{1}-z_{2}\neq\mathrm{i}\frac{2m\pi}{\gamma}, \forall m\in\mathbb{Z}$ in the form (iii). 

In the context of Floquet theory, the eigenvalues of the monodromy matrix $M$, denoted by $e^{z_{0}}$,  $e^{z_{1}}, e^{z_{2}}$ depending on the specific cases, are typically referred to as Floquet multipliers \cite{chicone2006ordinary}. Let $k_{0}=\frac{z_{0}}{\mathrm{i}}$, $k_{1}=\frac{z_{1}}{\mathrm{i}}$ and $k_{2}=\frac{z_{2}}{\mathrm{i}}$. $k_{0}$, $k_{1}$ and $k_{2}$ are what we called quasimomenta. 

Suppose $\psi$ is a solution of \cref{eq: sl} that is of form (i), given by $\psi(x)=e^{\mathrm{i}kx}p(x)$, where $k\in\mathbb{C}$ and $p$ is $\mathcal{R}$-periodic. For $\tilde{k}\in[k]_{\gamma}$,  there exists $\tilde{m}\in\mathbb{Z}$ such that 
\begin{equation*}
   k=\tilde{k}+\frac{2\tilde{m}\pi}{\gamma},
   \end{equation*}
hence 
\begin{equation*}
    \psi(x)=e^{\mathrm{i}kx}p(x)=e^{\mathrm{i}\tilde{k}x}e^{\mathrm{i}\frac{2\tilde{m}\pi}{\gamma}x}p(x). 
\end{equation*}

Since $e^{\mathrm{i}\frac{2\tilde{m}\pi}{\gamma}x}p(x)$ is also $\mathcal{R}$-periodic, $\tilde{k}$ is also the quasimomenta of $\psi$. Similarly, we are able to illustrate the quasimomenta of solutions of form (ii) or (iii).

We completes the proof.
\end{proof}

\begin{proof}[Proof of \cref{thm4.3}]
We obtain from \cref{thm4.1} that for $(\lambda, \psi)\in\mathbb{R}\times C_{p}^{2}(\mathbb{R})$, $\psi$ is of the three specific forms.
We then prove the conclusion according to the three different forms of the solutions.

When $\lambda\in\sigma_{g}^{1}(\mathscr{L})$, we substitute $\psi(x)=e^{\mathrm{i}k_{0}x}p(x)$ into \cref{eq: sl} and arrive at
\begin{equation}\label{eq4.8}
      \mathcal{D}^{k_{0}}p(x)=\lambda p(x), 
\end{equation}
which implies $p$ is the $C_{p, \gamma}^{2}(\mathbb{R})$ solution of $\mathcal{D}^{k_{0}}$. 

When $\lambda\in\sigma_{g}^{2}(\mathscr{L})$, we substitute $\psi(x)=e^{\mathrm{i}k_{0}x}(p_{1}(x)+xp_{2}(x))$ into \cref{eq: sl} and get
\begin{equation}\label{eq4.9}
     (\lambda \operatorname{id}-\mathcal{D}^{k_{0}})p_{1}(x)+x(\lambda \operatorname{id}-\mathcal{D}^{k_{0}})p_2(x)+(2\mathrm{i}k_{0}-W)p_{2}(x)+2p_{2}^{\prime}(x)=0,
\end{equation}
where $\operatorname{id}$ denotes the identity operator. Since $W, V, p_{1}$ and $p_{2}$ are $\mathcal{R}$-periodic, we substitute $x=x+\gamma$ into \cref{eq4.9} and get a new equation, which minus \cref{eq4.9} and finally yields
\begin{equation}\label{eq4.10}
    \mathcal{D}^{k_{0}}p_{2}(x)=\lambda p_{2}(x).
\end{equation}

Therefore,  $p_{2}$ is the $C_{p, \gamma}^{2}(\mathbb{R})$ solution of $\mathcal{D}^{k_{0}}$. 

We obtain from \cref{eq4.9} and \cref{eq4.10} that 
\begin{equation}\label{eq4.11}
    \mathcal{D}^{k_{0}}p_{1}(x)=\lambda p_{1}(x)+(2\mathrm{i}k_{0}-W)p_{2}(x)+2p_{2}'(x).
\end{equation}

When $\lambda\in\sigma_{g}^{3}(\mathscr{L})$, we substitute $\psi(x)=e^{\mathrm{i}k_{1}x}p_{1}(x)+e^{\mathrm{i}k_{2}x}p_{2}(x)$ into \cref{eq: sl} and arrive at
\begin{displaymath}
  e^{\mathrm{i}k_{1}x}(\mathcal{D}^{k_{1}}p_{1}(x)-\lambda p_{1}(x))+e^{\mathrm{i}k_{2}x}(\mathcal{D}^{k_{2}}p_{2}(x)-\lambda p_{2}(x))=0,
\end{displaymath}
which implies
\begin{equation}\label{eq4.12}
    e^{\mathrm{i}(k_{1}-k_{2})x}(\mathcal{D}^{k_{1}}p_{1}(x)-\lambda p_{1}(x))+(\mathcal{D}^{k_{2}}p_{2}(x)-\lambda p_{2}(x))=0.
\end{equation}

Since $p_{1}(x)$ and $p_{2}(x)$ are $\mathcal{R}$-periodic, there holds
\begin{equation}\label{eq4.13}
     e^{\mathrm{i}(k_{1}-k_{2})(x+l\gamma)}(\mathcal{D}^{k_{1}}p_{1}(x)-\lambda p_{1}(x))+(\mathcal{D}^{k_{2}}p_{2}(x)-\lambda p_{2}(x))=0,\quad \forall l\in\mathbb{Z}.
\end{equation}

Letting \cref{eq4.13} subtract \cref{eq4.12}, there holds for all $l\in\mathbb{Z}$ that 
\begin{equation}
    e^{\mathrm{i}(k_{1}-k_{2})x}(e^{\mathrm{i}(k_{1}-k_{2})l\gamma}-1)(\mathcal{D}^{k_{1}}p_{1}(x)-\lambda p_{1}(x))=0,
\end{equation}
which together with $[k_{1}]_{\gamma}\neq[k_{2}]_{\gamma}$ leads to
\begin{equation}\label{eq4.16}
    \mathcal{D}^{k_{1}}p_{1}(x)=\lambda p_{1}(x).
\end{equation}

By substituting \cref{eq4.16} into \cref{eq4.12}, we arrive at
\begin{equation}\label{eq4.17}
    \mathcal{D}^{k_{2}}p_{2}(x)=\lambda p_{2}(x).
\end{equation}
This completes the proof.
\end{proof}

We recall the following Liouville formula, which can be found in \cite{grimshaw1991nonlinear}.
\begin{lemma}\label{le:liou}
    Consider an $n$-dimensional first-order homogeneous linear differential equation
    \begin{equation}\label{eq:liou}
        y'=B(t)y, \quad t\in I,
    \end{equation}
    on an interval $I\subset\mathbb{R}$. Let $X$ be the fundamental matrix. If $\operatorname{tr} B(t)$ is a continuous function, then the determinant of $X$ satisfies, 
    \begin{displaymath}
       \det X(t)=\det X(t_{0})\exp\left(\int_{t_{0}}^{t}\operatorname{tr}B(s)ds\right)
    \end{displaymath}
    for any $t_{0},t\in I$.
    \end{lemma}

\begin{proof}[Proof of \cref{cor4.8}]    
According to \cref{le:liou}, let $X(x)$ be the fundamental function of \cref{eq4.2} and $B(x)=\left(\begin{array}{cc}
      0&1\\V-\lambda&W
    \end{array}\right)$. We obtain from \cref{le:liou} that,
\begin{equation*}
    \det X(\gamma)=\det X(0)\exp\left(\int_{0}^{\gamma}\operatorname{tr} B(s)ds\right)=e^{\gamma \bar{W}},
\end{equation*}
which together with \cref{eq3.5} implies
  \begin{equation*}
        \operatorname{det}M=e^{\gamma \bar{W}},
    \end{equation*}
    where $M$ denotes the monodromy matrix of \cref{eq4.2} with $W, V\in C_{p, \gamma}^{0}(\mathbb{R})$.

In accordance of the proof of \cref{thm4.1}, in the case of  \cref{item:eig-bloch} or case of  \cref{item:eig-growth}, we have
\begin{equation*}
    \det M=\det e^{\gamma C}=e^{2\mathrm{i}k_{0}\gamma}=e^{\gamma \bar{W}}, 
\end{equation*}
 which yields
 \begin{equation*}
     k_{0}=\frac{1}{2\mathrm{i}}\bar{W}+\frac{n\pi}{\gamma}, n\in\mathbb{Z}.
 \end{equation*} In the case of \cref{item:eig-comb-bloch}, similarly,  there hold 
  \begin{equation*}
      k_{1}+k_{2}=\frac{1}{\mathrm{i}}\bar{W}+\frac{2n\pi}{\gamma}, n\in\mathbb{Z}.
  \end{equation*}
\end{proof}

In accordance with the discussion of the 1-D Strum-Liouville problem, we have the following classification of solutions of 1-D periodic Schrödinger equations.
\begin{proposition}
Solutions of the 1-D Schrödinger equation \cref{eq:sco} with the potential $V\in C_{p, \gamma}^{0}(\mathbb{R})$ are of one of the following forms:
\begin{itemize}
    \item form (i), given by $\psi(x)=e^{\mathrm{i}\frac{n\pi}{\gamma}x}p(x), n\in\mathbb{Z}$, $|\mathcal{A}(\lambda)|=1$: a special class of Bloch solutions which are indeed periodic solutions i.e., $\psi(x+\gamma)=\psi(x)$ or anti-periodic solutions i.e., $\psi(x+\gamma)=-\psi(x), x\in\mathbb{R}$;
    \item form (ii), given by $e^{\mathrm{i}\frac{n\pi}{\gamma}x}(p_{1}(x)+xp_{2}(x))$, $|\mathcal{A}(\lambda)|=1$: represents a class of special solutions  that are not Bloch solutions;
    \item  form (iii), given by $e^{\mathrm{i}k_{1}x}p_{1}(x)+e^{\mathrm{i}k_{2}x}p_{2}(x)$ and $k_{1}+k_{2}=\frac{2n\pi}{\gamma}, n\in\mathbb{Z}$, $|\mathcal{A}(\lambda)|=2$: Bloch solutions or their linear combination whose quasimomenta are in two different congruence classes modulo $\frac{2\pi}{\gamma}$, i.e., $[k_{1}]_{\gamma}\neq[k_{2}]_{\gamma}$. In fact, if $\psi_{1}(x)=e^{\mathrm{i}k_{1}x}p_{1}(x)$ is a solution,  $\psi_{2}(x)=e^{\mathrm{i}k_{2}x}p_{2}(x)$ with $k_{2}=-k_{1}$ and $p_{2}(x)=p_{1}(-x)$ is also a solution corresponding to the same $\lambda$.
\end{itemize}
\end{proposition}

The following theorem  illustrates the relationship between bounded solutions and Bloch solutions. 
\begin{theorem}\label{pro}
The bounded solutions in the classical sense of the Schrödinger equation \cref{eq:sco} with the real valued potential $V\in C_{p,\gamma}^{0}(\mathbb{R})$ are either Bloch solutions or the linear combination of two Bloch solutions with real quasimomenta.
\end{theorem}
\begin{proof}
 We apply \cref{thm4.1} and \cref{thm4.3} to prove the conclusion. Let $\psi$ be an solution of \cref{eq:sco} corresponding to $\lambda$.
 
  If $\lambda\in\sigma_{g}^{2}(\mathscr{H})$, $\psi$ is of the form \cref{item:eig-growth} in \cref{thm4.1}, i.e., $\psi(x)=e^{\mathrm{i}k_{0}x}(p_{1}(x)+xp_{2}(x))$ with $p_{2}(x)\not\equiv0$, then the $\mathcal{R}$-periodic function $p_{2}\in C_{p}^{1}(\mathbb{R})$ implies that there exists a sequence $\{x_{n}\}_{n\geqslant1}\subset\mathbb{R}$ such that $x_n\to\infty$ and $|x_{n}p_{2}(x_{n})|\rightarrow\infty$ as $n\rightarrow\infty$. Indeed, since $p_{2}(x)\not\equiv0$, there exists $x_{0}$ such that $p_{2}(x_{0})\not=0$. Then, for $n=1, 2, \cdots$, let $x_{n}=x_{0}+n\gamma$ and we arrive at $|x_{n}p_{2}(x_{n})|=|x_{0}+n\gamma||p_{2}(x_{0})|\rightarrow\infty$ as $n\rightarrow\infty$.  Hence, the bounded solutions of \cref{eq:sco} must be of the forms \cref{item:eig-bloch} and \cref{item:eig-comb-bloch} in \cref{thm4.1}, 
i.e., the Bloch solutions or their linear combination.
 
Next it is sufficient to prove the conclusion for Bloch solutions.     Consider the Bloch solution $\varphi$, given by $\varphi(x)=e^{\mathrm{i}kx}p(x)$. Let $\operatorname{Im} k$ be the imaginary part of $k$. If $\operatorname{Im} k\neq0$, then 
 there exists a sequence $\{x_{n}\}_{n\geqslant1}$ such that $x_n\to \infty$ and $|\psi(x_n)|=|e^{\mathrm{i}kx_n}p(x_{n})|=e^{-\operatorname{Im}(k)x_{n}}|p(x_n)|\rightarrow\infty$. Hence, the bounded Bloch solutions of \cref{eq:sco} have quasimomentum $k_{0}\in\mathbb{R}$. 
\end{proof}

As mentioned in the introduction, the existence of Bloch solutions can be simply summarized as that the existence of nontrivial bounded solutions of \cref{eq:sco} implies the existence of Bloch solutions with real quasimomenta \cite{kuchment1993floquet, kuchment2016overview}. Indeed, \cref{pro} enhances that nontrivial bounded solutions are exactly the Bloch solutions or their linear combination with real quasimomenta in the one-dimensional case. Moreover, \cref{thm4.1}, \cref{thm4.3} and \cref{cor4.8} provide  a classification of solutions of 1-D Schr{\"o}dinger equations with periodic coefficients. We have classified the solutions according to the number of the congruence classes of the quasimomenta and showed the number is at most two.

For bounded Bloch solutions, quasimomenta are real. Note that in the case of $d$ dimensions, the Brillouin zone \cite{ashcroft1976solid, misra2011physics} $\mathcal{B}$ is isomorphic to the quotient space $\mathbb{R}^{d}/\mathcal{R}^{*}$ \cite{kuchment2016overview, lewin2010spectral}. According to the definitions, we see that the Brillouin zone $\mathcal{B}=[-\frac{\pi}{\gamma}, \frac{\pi}{\gamma}]$ for $d=1$. 
If we consider the quasimomenta restricted to $\mathcal{B}$, in other words, the congruence class  modulo $\frac{2\pi}{\gamma}$, we can draw the following conclusion with the classification of solutions of \cref{eq:sco}.

 \begin{corollary}\label{col:two}
   In the  Brillouin zone $\mathcal{B}$, there are at most two quasimomenta of Bloch solutions of \cref{eq:sco} corresponding to $\lambda$.
 \end{corollary}

\begin{corollary}\label{corres}
Restricted to the Brillouin zone $\mathcal{B}$, quasimomenta of solutions of \cref{eq:sco} corresponding to $\lambda$ are symmetric with respect to the origin.
\end{corollary}

For form \cref{item:eig-growth}, we have the following property.
\begin{theorem}\label{thm:pro1}
Suppose that the coefficient $V\in C_{p,\gamma}^{0}(\mathbb{R})$ and  $\psi(x)=e^{\mathrm{i}k_{0}x}(p_{1}(x)+xp_{2}(x)) $ with $p_{2}(x)\not\equiv0$ is a non-trivial solution of \cref{eq:sco}. Then $p_{1}(x)\equiv0$ if and only if the potential $V$ is constant, then  $[k_{0}]_{\gamma}=[0]_{\gamma}$.
\end{theorem}
\begin{proof}
It is obvious when $V$ is constant

Next let  the solution $\psi(x)$  correspond  $\lambda\in\mathbb{R}$ and $p_{1}(x)\equiv0$. \cref{eq4.11} implies 
\begin{equation}\label{eq:spec}
    2\mathrm{i}k_{0}p_{2}(x)+2p_{2}'(x)=0.
\end{equation}
And $p_{2}(x)=Ce^{-\mathrm{i}k_{0}x}$, where $C$ is some constant. Then we have the potential $V$ is a constant. 

Since $p_{2}(x)$ is $\mathcal{R}$-periodic, it holds that $p_{2}(x+\gamma)=p_{2}(x)$ and 
\begin{equation*}
    e^{-\mathrm{i}k_{0}\gamma}=1, 
\end{equation*}
which implies that $[k_{0}]_{\gamma}=[0]_{\gamma}$.
\end{proof} 

Furthermore, since quasimonenta of bounded Bloch solutions and the potential are real valued,  we obtain the following theorem that shows the property of periodic terms.

\begin{theorem}\label{thm:pro2}
   Suppose a bounded Bloch solution of the Schrödinger equation \cref{eq:sco} with  the  $C_{p,\gamma}^{0}(\mathbb{R})$ potential has a zero, then the quasimonentum is $\frac{n\pi}{\gamma}, n\in\mathbb{Z}$ and the solution equals the product of a complex constant and a real-valued $\gamma$-periodic or anti-periodic function. 
\end{theorem}

\begin{proof}
Suppose there exists $x_{0}\in\mathbb{R}$ such that the bounded Bloch solution $\psi(x)=e^{\mathrm{i}kx}p(x)$ of \cref{eq:sco} satisfies $\psi(x_{0})=0$. 
We obtain from \cref{thm4.3} that $p(x)$ is the solution of $\mathcal{D}^{k}$ for some $z\in\mathbb{C}$, which reads
\begin{displaymath}
  \mathcal{D}^{k}p(x)=\lambda p(x). 
\end{displaymath}
Equivalently, we have
\begin{equation}\label{eq4.19}
    (-k^{2}+\lambda-V)p(x)+2\mathrm{i}kp'(x)+p''(x)=0.
\end{equation}

We see from \cref{pro} that $\psi$ is bounded leads to $k\in\mathbb{R}$. Since $V$ and $\lambda$ are real-valued, by substituting $p(x)=a(x)+\mathrm{i}b(x)$ into the equation \cref{eq4.19} and subsequently separating the real and imaginary parts, we obtain:
\begin{align}
    (-k^{2}+\lambda-V)a-2kb'+a''=0,\label{eq4.20}\\
    (-k^{2}+\lambda-V)b+2ka'+b''=0.\label{eq4.21}
\end{align}

Multiply both sides of \cref{eq4.20} by $b$ and \cref{eq4.21} by $a$ and subtract one from the other, we get
\begin{equation*}\label{eq4.22}
    2k(aa'+bb')=a''b-ab'',
\end{equation*}
which yields
\begin{equation}\label{eq4.23}
    (k(a^{2}+b^{2}))'=(a'b-ab')'.
\end{equation}

By integrating both sides of \cref{eq4.23}, we arrive at
\begin{equation}\label{eq4.24}
    k(a^{2}+b^{2})=a'b-ab'+\tilde{C},
\end{equation}
where $\tilde{C}\in\mathbb{R}$ is some constant.

However, the condition $\psi(x_{0})=0$ implies that $a(x_{0})=b(x_{0})=0$. Hence, we obtain from \cref{eq4.24} that $\tilde{C}=0$ and then 
\begin{displaymath}
  k(a^{2}+b^{2})=a'b-ab',
\end{displaymath}
which leads to
 \begin{equation}\label{eq4.25}
   k((\frac{a}{b})^{2}+1)=\frac{a'b-ab'}{b^{2}}=(\frac{a}{b})'.
 \end{equation}
 
Note that the solution of \cref{eq4.25} is  $\frac{a}{b}=\tan(kx+C)$. We have
\begin{equation}\label{quionet}
        \frac{\operatorname{Re} p(x)}{\operatorname{Im} p(x)}=\tan(kx+C), 
\end{equation}
where $C\in\mathbb{R}$ is some constant and  $\operatorname{Re} p(x), \operatorname{Im} p(x)$ are real and imaginary parts of $p(x)$, respectively. Therefore, there exists a real valued function $f$ satisfying 
\begin{equation*}
    \operatorname{Re}p(x)=\sin(kx+C)f(x),\\
    \operatorname{Im}p(x)=\cos(kx+C)f(x),
\end{equation*}
or
\begin{equation}\label{eq:iiiii}
    p(x)=e^{\mathrm{i}(\frac{\pi}{2}-kx-C)}f(x).
\end{equation}
which yields
\begin{equation}\label{eq:ppppp}
    \psi(x)=\mathrm{i}e^{-\mathrm{i}C}f(x).
\end{equation}

Since $\psi\in C_{p}^{2}(\mathbb{R})$, we obtain from \cref{eq:ppppp} that $f\in C_{p}^{2}(\mathbb{R})$. Moreover, since $p$ is $\mathcal{R}$-periodic, we obtain from \cref{eq:iiiii} that
\begin{equation*}
    f(x+\gamma)=e^{\mathrm{i}k\gamma}f(x),\forall x\in\mathbb{R}.
\end{equation*}
Since $f$ is real valued, we have $e^{\mathrm{i}k\gamma}\in\mathbb{R}$, which implies that $k=\frac{n\pi}{\gamma}, n\in\mathbb{Z}$, and 
\begin{equation*}
    f(x+\gamma)=f(x), \forall x\in\mathbb{R}, 
\end{equation*}
or
\begin{equation*}
    f(x+\gamma)=-f(x), \forall x\in\mathbb{R}.
\end{equation*}
\end{proof}

\begin{remark}
Indeed, \cref{quionet} holds not only for Bloch solutions, but also for solutions of the form \cref{item:eig-growth} addressed in \cref{thm4.1}. It reads that, for the solution $\psi(x)=e^{\mathrm{i}kx}(p_{1}(x)+xp_{2}(x))$ with the condition $\psi(x_{0})=p_{1}(x_{0})$ for some $x_{0}\in\mathbb{R}$, it holds that 
\begin{equation*}
    \frac{\operatorname{Re} p_{2}(x)}{\operatorname{Im} p_{2}(x)}=\tan(kx+C),
\end{equation*}
where $C\in\mathbb{R}$ is some constant. Since $p_{2}(x)$ is also the  solution of $\mathcal{D}^{k}$ with \cref{thm4.3}, the proof is the same as that in \cref{quionet} except some notations.
\end{remark}

\bibliographystyle{siamplain}
\bibliography{references}
\end{document}